\documentclass[sigconf]{acmart}
\AtBeginDocument{%
  \providecommand\BibTeX{{%
    \normalfont B\kern-0.5em{\scshape i\kern-0.25em b}\kern-0.8em\TeX}}}

\acmYear{2022}\copyrightyear{2022}
\setcopyright{acmcopyright}
\acmConference[e-Energy '22]{The Thirteenth ACM International Conference on Future Energy Systems}{June 28--July 1, 2022}{Virtual Event, USA}
\acmBooktitle{The Thirteenth ACM International Conference on Future Energy Systems (e-Energy '22), June 28--July 1, 2022, Virtual Event, USA}
\acmPrice{15.00}
\acmDOI{10.1145/3538637.3538867}
\acmISBN{978-1-4503-9397-3/22/06}

\usepackage[framemethod=1]{mdframed}

\usepackage{color}
\definecolor{shadecolor}{rgb}{0.878906, 0.878906, 0.878906}
\usepackage{multirow}
\usepackage{float}
\usepackage{framed}
\usepackage{amsmath}
\usepackage{amsthm}
\usepackage{graphicx}
\usepackage{tabularx}

\usepackage{pifont}
\newcommand{\xmark}{\ding{55}}%

\makeatletter

\floatstyle{ruled}
\newfloat{algorithm}{tbp}{loa}
\providecommand{\algorithmname}{Algorithm}
\floatname{algorithm}{\protect\algorithmname}

\theoremstyle{plain}
\newtheorem{thm}{\protect\theoremname}
\theoremstyle{definition}
\newtheorem{defn}{\protect\definitionname}
\theoremstyle{plain}
\newtheorem{lem}{\protect\lemmaname}
\theoremstyle{plain}
\newtheorem{cor}{\protect\corollaryname}

\usepackage{graphics, subfig}
\usepackage{epsfig}
\@ifundefined{definecolor}
 {\usepackage{color}}{}
\usepackage{amsfonts}
\usepackage{latexsym}
\usepackage{tabularx}
\usepackage{dsfont}
\usepackage{comment}
\usepackage{colortbl}

\usepackage{algorithmic}
\floatname{algorithm}{Algorithms}

\usepackage{caption}
\ifodd1
\newcommand{\com}[1]{\textbf{\color{blue} (COMMENT: #1)}}

\else\newcommand{\com}[1]{}\fi

\newcommand{\beq}{\begin{equation}}\newcommand{\eeq}{\end{equation}}\newcommand{\bea}{\begin{eqnarray}}\newcommand{\eea}{\end{eqnarray}}\newcommand{\bda}{\begin{eqnarray*}}\newcommand{\eda}{\end{eqnarray*}}\newcommand{\bdalign}{\begin{align*}}\newcommand{\edalign}{\end{align*}}

\@ifundefined{showcaptionsetup}{}{%
 \PassOptionsToPackage{caption=false}{subfig}}
\usepackage{subfig}
\makeatother

\usepackage[english]{babel}
\providecommand{\corollaryname}{Corollary}
\providecommand{\definitionname}{Definition}
\providecommand{\lemmaname}{Lemma}
\providecommand{\theoremname}{Theorem}

\begin{document}


\title{Competitive Prediction-Aware Online Algorithms for Energy Generation Scheduling in Microgrids}


\author{Ali Menati}
\affiliation{%
   \institution{School of Data Science}
   \country{City University of Hong Kong}}

   \author{Sid Chi-Kin Chau}
\affiliation{%
   \institution{Research School of Computer Science}
   \country{Australian National University}}
   
      \author{Minghua Chen}
\affiliation{%
   \institution{School of Data Science}
   \country{City University of Hong Kong}}

\authornote{Corresponding  author: Minghua Chen.}

\begin{abstract}
 Online decision-making in the presence of uncertain future information is abundant in many problem domains. In the critical problem of energy generation scheduling for microgrids, one needs to decide when to switch energy supply between a cheaper local generator with startup cost and the costlier on-demand external grid, considering intermittent renewable generation and fluctuating demands. Without knowledge of future input,  competitive online algorithms are appealing as they provide optimality guarantees against the optimal offline solution. In practice, however, future input, e.g., wind generation, is often predictable within a limited time window, and can
be exploited to further improve the competitiveness of online algorithms. In this paper, we exploit the structure of information in the prediction window to design a novel prediction-aware online algorithm for energy generation scheduling in microgrids. Our algorithm achieves the best competitive ratio to date for this important problem, which is at most $3-2/(1+\mathcal{O}(\frac{1}{w})),$ where $w$ is the prediction window size. We also characterize a non-trivial lower bound of the competitive ratio and show that the competitive ratio of our algorithm is only $9\%$ away from the lower bound, when a few hours of prediction is available.  Simulation results based on real-world traces corroborate our theoretical analysis and highlight the advantage of our new prediction-aware design.
\end{abstract}


\settopmatter{printfolios=true}
\maketitle

\pagestyle{plain}
\thispagestyle{plain}


\section{Introduction}
Central to online decision-making problems is the presence of future information, which, if available, determines the optimal decisions taken currently. Without knowledge of future information, competitive online algorithms are robust decision-making algorithms that can offer a worst-case guarantee to their sub-optimal decisions,  against the optimal offline decisions with complete future information. In microgrids, it is essential to serve the fluctuating demands using local generators, intermittent renewable energy sources, and an external grid with time-varying tariffs. This is a well-studied class of problems in smart grid literature (including economic dispatching \cite{gaing2003particle} and unit commitment problems \cite{kazarlis1996genetic}). It is appealing to employ competitive online algorithms for efficient energy management in microgrids \cite{narayanaswamy2012online}. In practice, however, prediction is often plausible within a limited time window. For example, in smart grid, the advances of machine learning and big data analytics enable relatively accurate renewable energy forecasting with several hours ahead \cite{Wang2018MultistepAW}. The availability of predicted future information will certainly enhance the design of online decision-making algorithms, providing the missing information for optimal decisions. Nonetheless, prediction is never perfect. When we consider only a limited time window of accurate prediction, it may not be sufficient to determine the current optimal decisions. Thus, a worst-case guarantee is still desirable to benchmark an online algorithm's sub-optimal decisions using limited future information against the optimal decisions with complete future information.

In this paper, a novel prediction-aware online algorithm is provided for energy generation scheduling in microgrids that considers a prediction window. We note that it is non-trivial to design a competitive prediction-aware online algorithm. A straightforward approach is to use receding horizon control (RHC), which determines the best possible decisions based on only the predicted future information, but does not consider any future events beyond the prediction window. Hence, RHC is not robust against the uncertainty beyond prediction. Therefore, a more robust algorithm is needed that can both harness the predicted information and accommodate the uncertainty beyond prediction. Such an algorithm should be sufficiently general to consider a parameterized prediction window with any window size. In this paper, we consider the energy generation scheduling problem for microgrids, where one needs to decide when to switch energy supply between a cheaper local generator with startup cost and the costlier on-demand external grid, considering intermittent renewable generation and fluctuating demands. There have been a number of recent studies about online energy generation scheduling. In the previous study \cite{Minghua2013SIG}, a prediction-oblivious algorithm called \textsf{CHASE}
has been proposed to solve this problem. It is shown that \textsf{CHASE} achieves a competitive ratio of $3$, which is the best among all deterministic online algorithms.  More generally, there is an abstract framework called Metrical Task System (MTS) problem \cite{borodin2005online}, which considers general online decision-making processes for state changes with uncertain future switching costs among the states. We note that the online energy generation scheduling problem belongs to a class of scalar MTS problems, where the states are the number of generators being on (or off). However, there is no prediction-aware online algorithm for MTS in the literature so far, to the best of our knowledge. In this paper, we present a novel prediction-aware online algorithm with the best competitive ratio to date; see Sec.~\ref{sec:relwork} for the discussion. Our algorithm not only solves the online energy generation scheduling problem, but also paves the way of tackling more general MTS problems with limited predicted information. MTS is capable of modeling many problems arising in a wide range of applications, including embedded systems and data centers \cite{DBLP:journals/tecs/IraniSG03},  transportation systems \cite{coester2018online}, and online learning \cite{Blum97on-linelearning}. As another novelty considered in this work, the previous study  \cite{Minghua2013SIG} focuses on a homogeneous setting of local generators. In practice, however, microgrids may employ different types of generators with heterogeneous operating constraints. In this paper, we consider a more general setting where local generators can be heterogeneous with different capacities. We summarize our main contributions as follows:	\vspace{-1.5mm} \begin{enumerate}
\item We propose \textsf{CHASEpp} as a novel prediction-aware online algorithm that can further improve the competitive ratio of the state-of-the-art \textsf{CHASElk}. This algorithm achieves competitive ratio of $ 3-\big(2\alpha +2(1-\alpha)/(1+\mathcal{O}(\frac{1}{w}) \big) \leq  3-2/(1+\mathcal{O}(\frac{1}{w})) $, where $\alpha \in [0,1]$ is the system parameter that captures price discrepancy between using local generation and external sources to supply energy. Our algorithm achieves the best competitive ratio to date with up to $20\%$ improvement than the state-of-the-art \textsf{CHASElk}. This competitive ratio also decreases twice faster with respect to $w$ than \textsf{CHASElk}.  We explore a new design space in our algorithm called cumulative differential cost in the prediction window, to better utilize the prediction information in making more competitive decisions. Our approach proactively monitors the possible online to offline cost ratio in the prediction window and makes intelligent online decisions. 
\item
We also characterize a non-trivial lower bound of the competitive ratio. To obtain the lower bound, we create an adversary that progressively generates a worst-case input for any algorithm. We assume at anytime the accurate prediction of a future window is available to the algorithm. This means the adversary needs to build a window of input ($\big[\sigma(\tau)\big]_{t}^{t+w}$) without knowing the algorithm's behavior in the upcoming window, which makes it difficult to establish the lower bound. In Sec.~\ref{sec:exp}, we show that the competitive ratio of \textsf{CHASEpp} is close to the lower bound. For example, they only differ by $9\%$ (i.e., 1.94 vs. 1.75) when we have a few hours of predictions.

\item
 We use both theoretical analysis and trace-driven experiments to evaluate the performance of our algorithm by comparing it with the state-of-the-art algorithms. We also elaborate on how the input structure and system parameters can affect its performance. We note that our approach, with both perfect and noisy prediction information, can be extended to the online algorithm design for a general class of MTS problems with a similar structure.

\end{enumerate}

\section{Related Work}
\label{sec:relwork}
Generation scheduling problems have attracted considerable attention. In large-scale power systems, with high aggregation effect on demand and small percentage of the erratic renewable generation, predicting the demand in the entire time horizon with a good level of accuracy is possible. Therefore the energy generation scheduling is basically an offline problem. Two main forms of this problem are economic dispatching  \cite{chen1995large,selvakumar2007new,gaing2003particle} and unit commitment \cite{kazarlis1996genetic,guan2003optimization,padhy2004unit}. In the literature, researchers tackled this problem in different ways including dynamic programming \cite{snyder1987dynamic}, stochastic programming \cite{takriti1996stochastic}, and mixed-integer programming 
\cite{carrion2006computationally}. In recent years, with the increasing integration of the highly fluctuating renewable sources and deployment of small-scale microgrids, the uncertainty and intermittency have increased substantially on both supply and demand sides, and local supply-demand matching has become an essential part of the microgrid operation. Therefore, the previous approaches for the traditional grid are not applicable to this new scenario, where we do not know all the information on the time horizon \cite{kroposki2017integrating, yang2018economic}.
The microgrid operator is the responsible party for local power balancing, which determines the optimal power generation and scheduling of all on-site
resources \cite{microgrid}. To address the supply-demand matching problem in microgrids, researchers proposed different approaches which aim at scheduling either dispatchable generation on the supply side \cite{narayanaswamy2012online} or flexible load on the demand side (e.g., \cite{chang2013real, huang2014adaptive}). Some other works combined these two with energy storage management \cite{chen2013heterogeneous, guo2013decentralized}, in order to achieve power balance in microgrid.

In recent years, online optimization has emerged as a foundational topic in a variety of computer systems. It has seen considerable attention from applications in a wide range of research, including networking and distributed systems \cite{tu2013dynamic,urgaonkar2011optimal,neely2010stochastic}.      
\begin{table}
    \scriptsize
	\begin{center}
		\begin{tabular}{|p{1.4cm}||c|c|c|c|} 
			\hline  \textbf{Reference} &  \vtop{\hbox{\strut \textbf{Structure} }\hbox{\strut \textbf{Exploitation}}}   & \vtop{\hbox{\strut \textbf{Competitive} }\hbox{\strut \textbf{Ratio}}} & \vtop{\hbox{\strut \textbf{Lower} }\hbox{\strut \textbf{Bound}}}  & \vtop{\hbox{\strut \textbf{Heterogeneous} }\hbox{\strut \textbf{Generators}}}  \\ 
			\hline\hline  Lin \textit{et al.}  \cite{ocoprediction} & \xmark   & \vtop{\hbox{\strut arbitrarily large (in a  }\hbox{\strut more general setting)}} & \xmark  & \xmark  \\ 
			\hline  Hajiesmaili \textit{et al.} \cite{hajiesmaili2016rand} & \xmark  & heuristic & \xmark  & \xmark  \\ 
			\hline  Lu \textit{et al.} \cite{Minghua2013SIG} & \xmark  & \vtop{\hbox{\strut sub-optimal (partial }\hbox{\strut use of the information)}}  & \xmark  & \xmark  \\ 
			\hline\hline  \textbf{This work} & \checkmark  & \vtop{\hbox{\strut reduces twice faster }\hbox{\strut than \cite{Minghua2013SIG} with $w$}   } & \checkmark  & \checkmark  \\ 
			\hline 
		\end{tabular} 
		\caption{Summary and comparison of existing works.}\vspace{-10pt}
		\label{tbl:sum}
	\end{center}
	\vspace{-7mm}
\end{table}   
In \cite{narayanaswamy2012online}, online convex optimization (OCO) framework \cite{zinkevich2003online} is used to design algorithms for the microgrid economic dispatch. OCO is a prominent paradigm being increasingly applied in different applications \cite{caoVirtual2018,Cao2018OnTT,Chen2016Prediction,Chen2015OCO, 6322266, 8486362}. There are some similarities between OCO with switching cost for dynamic scaling in datacenters \cite{5934885} and the one of energy generation \cite{Minghua2013SIG}. However, the inherent structures of both problems and solutions are significantly different. First, OCO considers a continuous feasible region, while in our setting with energy generators, the decision variable can only take discrete values. Second, their solution only applies to the homogeneous setting, while our solution can utilize multiple heterogeneous local generation units. Finally, in their recent work \cite{ocoprediction}, the competitive ratio is sub-optimal compared to the existing solutions and our new algorithm. By increasing the switching cost, their competitive ratio increases linearly, while our algorithm's competitive ratio is always upper bounded by a constant that is independent of the switching cost. Other prediction-aware online algorithms like the one in \cite{li2018online} also produce competitive ratios that grow unbounded as the switching cost increases. Some recent works \cite{Xiaojun2019, Shi2021CombiningRW} tried to solve this issue by designing online algorithms with bounded competitive ratios. Still, their algorithm can only leverage prediction for large enough window sizes $w \geq r_{co}$, where $r_{co}$ is a constant that grows unbounded as the switching cost increases. For example, in a real-world setting used in our numerical experiments (Sec.~\ref{sec:exp}), for $w<11$ hours, their competitive ratio is a constant that is independent of the window size. Meanwhile, the competitive ratio of our algorithm always keeps decreasing as the window size increases. In \cite{zhang2015peak}, a competitive algorithm design approach is used to solve the online economic dispatching problem with a peak-based charging model, which does not take the startup cost into account. The study in \cite{Minghua2013SIG} incorporates the startup cost and turns the problem into a joint unit commitment and economic dispatch problem. In \cite{hajiesmaili2016rand}, a randomized online algorithm is proposed to solve this problem. In this paper, we aim to solve this problem with accurate prediction of the near future demand. In \cite{Minghua2013SIG}, a prediction-aware online algorithm has been proposed to this end, but it fails to utilize all the given predicted information. Here we propose a novel competitive online algorithm that will further improve both theoretical and practical performance over the previous algorithm. There are several aspects both in algorithm design and theoretical analysis that make our work to be different from other online solutions. We compare the most important aspects of these works and our work in Table~\ref{tbl:sum}.  
  \vspace{-1mm}

\section{Energy Generation Scheduling Problem}
\label{sec:pf}
The objective of energy generation scheduling in microgrids is to coordinate various energy sources such as local generation units and renewable sources to fulfill both electricity and heat demands while minimizing the total energy cost. It can be formulated as a microgrid cost minimization problem (\textsf{MCMP}). We consider a system that operates in a time-slotted fashion, where $\mathcal{T}$ is the set of time slots, and the total length of the time horizon is $T$ time slots ($T\triangleq |\mathcal{T}|$). 
The key notations are presented in Table~\ref{tbl:not}.
	\vspace{-1mm}
\subsection{\label{sec:sysmod}System Model}
\textbf{Energy demand}: 
The energy demand profile includes two types of energy demand, namely electricity demand and heat demand. Let $a(t)$ and $h(t)$ be the net electricity demand  (i.e., the residual electricity demand not covered by renewable generation) and the heat demand at time $t$, respectively.

\textbf{External grid and heating}: We assume the microgrid operates in the  ``grid-connected'' mode, and the unbalanced electricity demand can be acquired from the external grid in an on-demand manner. We denote $p(t)$ as the spot price from the electricity grid at time $t$, where $p(t)\in [P_{\mathrm{min}},P_{\mathrm{max}}]$. To keep the generality of the problem, we do not assume any specific stochastic model for the input profile $\sigma(t) \triangleq (a(t), h(t), p(t))$. Finally, to cover the heating demand, we can use external natural gas, costing $c_{\mathrm{g}}$ per unit of demand.

\textbf{Local generation}:
We consider a heterogeneous setting where the power output capacity for power generation $n\in [1, N]$ is $L_n$, and these capacities can be different from each other. This generalizes the homogeneous setting considered in \cite{Minghua2013SIG}, where local generators have identical capacities. By adopting the widely-used generator model \cite{kazarlis1996genetic}, We denote $\beta$ as the startup cost of turning on a generator, $c_{\mathrm{m}}$ as the sunk cost per unit time of running a generator in its active state, and $c_{\mathrm{o}}$ as the incremental operational cost per unit time for an active generator to output one unit of energy. In a more realistic model of generators, two additional operational constraints are considered. Namely,  minimum turning on/off periods, and ramping up/down rates. In \cite{Minghua2013SIG}, a general problem that includes these additional constraints is considered, and the approach to solve them is also proposed. In this paper, we focus on the ``fast-responding'' generators whose minimum on/off period constraint and ramping-up/down constraint is negligible.  Our solution can then be extended to the case with general generators using the same approach as in \cite{Minghua2013SIG}. Finally, we assume the local generators are CHP generators that can generate both electricity and heat simultaneously. We denote  $\eta$ as the heat recovery efficiency for co-generation  \textit{i.e.}, for each unit of electricity generated, $\eta$ unit of useful heat can be supplied for free. Thus, $\eta c_{\mathrm{g}}$ is the cost-saving due to using co-generation to supply heat, provided that there is sufficient heat demand.  Note that by setting $\eta=0$, the problem reduces to the case of  a system with no co-generation. We assume $c_{\mathrm{o}} \geq \eta \cdot c_{\mathrm{g}} $, which means it is cheaper to obtain heat by using natural gas than purely by generators. To keep the problem interesting, we assume that $c_{\mathrm{o}}+ \frac{c_{\mathrm{m}}}{L} \leq p_{\mathrm{max}}+ \eta \cdot c_{\mathrm{g}} $. This assumption ensures that the minimum co-generation energy cost is cheaper than the maximum external energy price. If this assumption does not hold, the optimal decision is to always acquire power and heat externally and separately. In this paper, we do not consider 
using energy storage in the generation scheduling problem. The reason is that for the typical size of microgrids, e.g., a college campus, existing energy storage systems are rather expensive and not widely available \cite{menati2021preliminary}. 
\vspace{-1mm}
\begin{table}[!t]
	\begin{center}
		\begin{tabular}{|c|c|p{6.5cm}|}
			\hline
			\multicolumn{2}{|c|}{\textbf{Notation}} & \textbf{Definition} \\
			\hline \hline
			\multirow{10}{*}{\rotatebox[origin=c]{90}{\hspace{10mm} \textbf{Generator} }} 			
			&$\beta$  & The startup cost of local generator (\$)\tabularnewline
			&$c_{m}$  & The sunk cost per interval of running local generator (\$)\tabularnewline
			&$c_{o}$  & The incremental operational cost per interval of running local generator
			to output an additional unit of power
			(\$/Watt)\tabularnewline
			&$L$  & The maximum power output of generator (Watt)\tabularnewline
			&$\eta$ & The heat recovery efficiency of co-generation \tabularnewline\hline
			\multirow{12}{*}{\rotatebox[origin=c]{90}{ \hspace{9mm} \textbf{Demand}}}&$\mathcal{T}$  & The set of time slots ($T\triangleq |\mathcal{T}|$)\tabularnewline 
			&$c_{g}$ & The price per unit of heat obtained externally using natural gas
			(\$/Watt)\tabularnewline
			&$a(t)$  & The net electricity demand minus the instantaneous renewable supply at time $t$ (Watt)\tabularnewline
			&$h(t)$ & The heat demand at time $t$ (Watt)\tabularnewline
			&$p(t)$  & The spot price per unit of power obtained from the electricity grid
			($P_{\min}\leq p(t)\leq P_{\max}$) (\$/Watt)\tabularnewline
			&$\sigma(t)$ & The joint input at time $t$: $\sigma(t) \triangleq (a(t), h(t), p(t))$ \tabularnewline\hline
			\multirow{7}{*}{\rotatebox[origin=c]{90}{ \hspace{4mm} \textbf{Opt. Var}}} &			$y(t)$  & The on/off status of the local generator (on as ``$1$'' and off as ``$0$'') \tabularnewline
			&$u(t)$  & The power output level when the generator is on (Watt) \tabularnewline
			&$s(t)$ & The heat level obtained externally by natural gas (Watt)\tabularnewline
			&$v(t)$  & The power level obtained from electricity grid (Watt)\tabularnewline\hline
		\end{tabular} 
	\end{center}
	\caption{Key Notations. Brackets indicate the unit. We denote a vector by a single symbol, {\em e.g.,} ${a\triangleq\big[a(t)\big]_{t=1}^{T}}$. }
	\label{tbl:not}
	\vspace{-5mm}
\end{table}
\subsection{Problem Formulation}  \label{ssec:problem_definition}   
Let $v(t)$ and $s(t)$ be the amount of electricity and heat obtained from the external grid and the external natural gas at time $t$, respectively. Let $y_{\mathrm{n}}(t)$ be the generator binary on/off status ($1$ as on and $0$ as off), and $u_{\mathrm{n}}(t)$ be the power output level of the $n$-th generator. The microgrid aggregated operational cost over the time horizon $\mathcal{T}$ is given by
\begin{eqnarray} \label{mcmpproblem}
	& \textsf{cost}(y,u,v,s)  \triangleq  \sum_{t\in\mathcal{T}}\Big( p(t)v(t)+c_{\mathrm{g}}s(t)+ \\
	& \sum_{n=1}^{N}[c_{\mathrm{o}}u_{\mathrm{n}}(t)+ c_{\mathrm{m}}y_{\mathrm{n}}(t)+\beta[y_{\mathrm{n}}(t)-y_{\mathrm{n}}(t-1)]^{+} ]   \Big), \notag
\vspace{-7mm}
\end{eqnarray}

that includes the grid electricity, external gas costs, and cost of the local generators, which is calculated by adding their operational cost and their switching cost over the entire time horizon $\mathcal{T}$. In this paper, we assume the initial status of all generators is off, \textit{i.e.,} $y_{\mathrm{n}}(0) = 0$. Given the cost function and decision variables we formulate the \textbf{Microgrid Cost Minimization Problem} (\textsf{MCMP}) as follows: 
\begin{subequations} \label{prob2}
	\begin{eqnarray} 
	&\underset{{y,u,v,s}}{\min}& \textsf{cost}(y,u,v,s) \\
	&\mbox{s.t.}\;& 0 \leq u_{\mathrm{n}}(t)\leq L_{\mathrm{n}} y_{\mathrm{n}}(t), \label{C_max_output}\\
	&& \textstyle{\sum}_{n=1}^N u_{\mathrm{n}}(t)+v(t) = a(t),   \label{C_e-demand}\\
	&& \eta \cdot \textstyle{\sum}_{n=1}^N u_{\mathrm{n}}(t)+s(t)\geq h(t),  \label{C_h-demand}\\
	&\mbox{vars.}\;& y_{\mathrm{n}}(t)\in \{0,1\},u_{\mathrm{n}}(t),v(t),s(t)\in \mathbb{R}_0^{+},  n\in [1,N], t\in \mathcal{T}, \nonumber
	\end{eqnarray} 
\end{subequations}
where the constraint~\eqref{C_max_output} captures the capacity limit of the generators and the constraints~\eqref{C_e-demand}-\eqref{C_h-demand} assure the electricity and heat demands are covered using the grid, natural gas, and the generators. It should be noted that the constraint~\eqref{C_e-demand} is in the form of equality, which means that the electricity power-balance constraint is strictly satisfied. To ensure this, we run the local CHP generators, which might produce a heating supply that is more than the demand. There are various mechanisms to manage the excessive generated heat, including thermal storage systems coupled with CHP units, which allow storing energy and reusing it later by lowering the temperature of a substance, such as water \cite{chpheat}.

We note that the AC Optimal Power Flow (OPF) constraints in the microgrid are not considered in \textsf{MCMP}, which is a joint unit commitment and economic dispatch problem. For certain microgrids, the ACOPF and economic dispatch should be coupled. However, if the microgrid is relatively large, similar to the conventional electric grids, to reduce the computational complexity, first unit commitment and economic dispatch are solved in hour-ahead time-scales, and then optimal power flow is solved minutes ahead of real-time \cite{OPF}. On the other hand, for relatively small-scale microgrids, because of the short distances, negligible losses, and large line capacities, the constraints of the ACOPF problem will not be activated, and usually, the generation cost has the dominant impact on microgrid planning. Therefore, although in some microgrids with fast-responding generators, the economic dispatch and the ACOPF are solved together, for relatively small or relatively large microgrids, this is not the case. In a general setting, the minimum turning on/off periods, and the ramping up/down rates can also be formulated as additional constraints. In \cite{Minghua2013SIG} the authors propose an approach, which obtains the solution to the general problem using the solution to the ``fast-responding'' generators setting. A simple heuristic is to first compute solutions using the online and offline algorithms without the constraints and then modify the solutions to respect the switching constraints. In this paper, we also focus on the ``fast-responding'' generators, but our offline and online algorithms can be easily updated to incorporate the switching constraints of the general case using the same approach. Note that this minimization problem is challenging to solve for several reasons. First, even in offline setting, this problem is a mixed-integer linear problem, which is generally difficult to solve. Second, the startup cost $\beta[y_{\mathrm{n}}(t)-y_{\mathrm{n}}(t-1)]^+$ term in the objective function makes decisions coupled across the time, hence the problem cannot be decomposed. Finally, the input profile $\sigma(t) \triangleq (a(t), h(t), p(t))$ arrives online and we may not know the complete future input. In this paper, we first consider a microgrid with a single generator and solve the \textsf{MCMP}. Later in Sec.~\ref{sec:multigenerator}, we extend the solution to the multiple generator scenario. Therefore, we drop the subscript $n$, and the problem $\textsf{MCMP}$ reduces to the problem $\textsf{MCMP}_{\mathrm{s}}$ for a single generator. We also utilize a useful observation to simplify the formulation: if the on/off status is given, the startup cost is determined, and $\textsf{MCMP}_{\mathrm{s}}$  reduces to a timewise decoupled linear program. According to \cite{Minghua2013SIG} given a fixed on/off status $\big(y(t)\big)_{t=1}^{T}$, the solution that minimizes $\textsf{cost}(y,u,v,s)$ is
	\begin{equation} 
		u(t)\mbox{=}\begin{cases}
			0, & \mathrm{if } \quad p(t)+\eta \cdot c_{\mathrm{g}}\leq c_{\mathrm{o}},\\
			\min\Big\{\frac{h(t)}{\eta},a(t),L y(t)\Big\}, &  \mathrm{if } \quad p(t)<c_{\mathrm{o}}<p(t)\mbox{+}\eta \cdot c_{\mathrm{g}},\label{eq:optimal_u}\\
			\min\Big\{a(t),L y(t)\Big\}, &  \mathrm{if } \quad c_{\mathrm{o}}\leq p(t), \notag
		\end{cases}
	\end{equation}
	\begin{equation}
	 v(t)=\left[a(t)-u(t)\right]^{+}, \,	\mathrm{and } \quad s(t)=\left[h(t)-\eta \cdot u(t)\right]^{+}.
	 \label{lem:fMCMP}
	\end{equation}

By using~\eqref{lem:fMCMP}, the problem $\textsf{MCMP}_{\mathrm{s}}$ can be further simplified to the following problem with a single decision variable to turn on ($y(t) = 1$) or off ($y(t) = 0$) the generator. 
\begin{eqnarray}
\textsf{MCMP}_{\mathrm{s}}: & \underset{y}{\min} \; \sum_{t\in\mathcal{T}}\Big\{ \psi\big(\sigma(t),y(t)\big) +\beta[y(t)-y(t-1)]^{+}\Big\} \notag\\
	& \mbox{vars.}\;\; y(t)\in \{0,1\}, t\in \mathcal{T}, \notag
\end{eqnarray}
where $\psi\big(\sigma(t),y(t)\big)\triangleq p(t)v(t)+c_{\mathrm{g}}s(t)+c_{\mathrm{o}}u(t)+ c_{\mathrm{m}}y(t)$, and $(u(t),v(t),s(t))$ are defined based on the result in~\eqref{lem:fMCMP}.
 
\section{Preliminary on Offline and Online Solutions}
\label{sec:chase}
We first review state-of-the-art online solutions and the optimal offline solution for $\textsf{MCMP}_{\mathrm{s}}$, providing necessary understandings towards designing a new algorithm later.

\subsection{Optimal Offline Algorithm Design}
In the offline setting the input $\big[\sigma(t)\big]_{t=1}^{T}$  is given at the beginning. We define
\begin{equation}
	\delta(t)\triangleq\psi(\sigma(t),0)-\psi(\sigma(t),1),\label{eqn:delta-definition}
\end{equation} 
to capture the single-slot cost difference between using or not using the local generation. When $\delta(t)>0$ (resp. $\delta(t)<0$), we tend to turn on (resp. off) the generator. However, to avoid turning on/off the generator too frequently, the cumulative cost difference $\Delta(t)$ is defined as
\begin{equation}
	\Delta(t)\triangleq\min\Big\{0,\max\{-\beta,\Delta(t-1)+\delta(t)\}\Big\},\label{eqn:Delta-definition}
\end{equation}
where the initial value is $\Delta(0)=-\beta$. Now that $\Delta(t)$ is defined, we divide the time horizon $\mathcal{T}$ into several disjoint sets called critical segments. As shown in Fig.~\ref{fig:An-example-of-CHASE}, each critical point $T_{i}^{c}$ is defined using a pair $(T_{i}^{c},\tilde{T_{i}^{c}})$ corresponds to an interval where $\Delta(t)$ goes from -$\beta$ to $0$ or from $0$ to -$\beta$, without touching the boundaries.
Based on the boundary values of these critical segments, we classify them into four categories as follows: 
\begin{itemize}
		\item \textbf{type-start}: $[1,T_{1}^{c}]$
		\item \textbf{type-$1$}: $[T_{i}^{c}+1,T_{i+1}^{c}]$, if $\Delta(T_{i}^{c})=-\beta$
		and $\Delta(T_{i+1}^{c})=0$
		\item \textbf{type-$2$}: $[T_{i}^{c}+1,T_{i+1}^{c}]$, if $\Delta(T_{i}^{c})=0$
		and $\Delta(T_{i+1}^{c})=-\beta$
		\item \textbf{type-end}: $[T_{k}^{c}+1,T]$.	
\end{itemize} 
In \cite{Minghua2013SIG}, the optimal offline solution of $\textsf{MCMP}_{\mathrm{s}}$ is given by
	\begin{equation}\label{thm:OFA-optimal}
		y^{\star}(t)\triangleq
		\begin{cases}
			1, & \text{if $t$ is in a type-1 segment},\\  
			0, & \text{otherwise}.
		\end{cases}
	\end{equation}

After getting the generator on/off status $y^{\star}(t)$, we apply~\eqref{lem:fMCMP} to obtain the optimal $u(t)$, $v(t)$, and $s(t)$.
\subsection{Online Algorithm Without Prediction} \label{onlinealg}
To evaluate the performance of the online algorithm, the competitive ratio is defined as follows:
\begin{defn}
    Let $\mathcal A$ be an online algorithm for $\textsf{MCMP}_{\mathrm{s}}$. Define
    \begin{equation}
    {\sf CR}(\mathcal A)\triangleq \max\limits_{\sigma}  \frac{{\rm Cost}({\sf y}_{\mathcal A})}{{\rm Cost}({\sf y}_{\rm OFA})}.
    \end{equation}  
\end{defn}
It is the worst-case ratio of the online cost over the offline cost.
We proceed by explaining the online algorithm \textsf{CHASE} \cite{Minghua2013SIG} that is later used in designing our new prediction-aware online algorithm. Recall that in the offline setting, we can detect the beginning of each critical segment right after the process enters them and set $y(t)$ accordingly. However, in the online setting, with no future information, it is impossible to do so. But, as shown in Fig.~\ref{fig:An-example-of-CHASE}, when $\Delta(t)=0$ (resp. $\Delta(t)=-\beta$), we are sure that we entered a type-1 (resp. type-2). Hence, $\textsf{CHASE}$ sets $y(t)=1$ (resp. $y(t)=0$). 
Intuitively, $\textsf{CHASE}$ tracks the offline optimal in an online manner, and its competitive ratio satisfies
\begin{equation}
    {\sf CR}(\textsf{CHASE}) \leq 3-2\alpha,
    \label{thm:chase}
\end{equation}
where
\begin{equation}
  \alpha  \triangleq \frac{c_{\mathrm{o}}+c_{\mathrm{m}}/L}{p_{\mathrm{max}}+\eta c_{\mathrm{g}}} \in (0,1]
  \label{def:alpha}
\end{equation}
and no other deterministic online algorithm can achieve a better competitive ratio. The adversarial input $\big[\sigma(t) \triangleq (a(t), h(t), p(t)) \big]_{1}^{T}$ that results in this worst-case competitive ratio for $\textsf{CHASE}$ is the input that always tries to make the online algorithm incur the maximum cost. Therefore, when the generator is on, the adversary gives zero demand $a(t)=0$, and when the generator is off, the adversary gives maximum demand $a(t)=L$ as the input, as follows:
\begin{eqnarray}
a(t)=
\begin{cases}
			L, & \mathrm{if } \, \, \,   y(t-1)=0 ,\\
			0, &  \mathrm{if } \, \, \,   y(t-1)=1,
		\end{cases}
\, \, \, \, h(t)= \eta a(t), \, \textit{and} \, \,\, p(t)= p_{\mathrm{max}}.
\end{eqnarray}
Later in Sec.~\ref{subsec:worst}, when analyzing the competitive ratio of our new prediction-aware online algorithm, we also present its corresponding worst-case input.

\begin{figure}[t]
	\centering{}\includegraphics[ width=0.85\columnwidth]{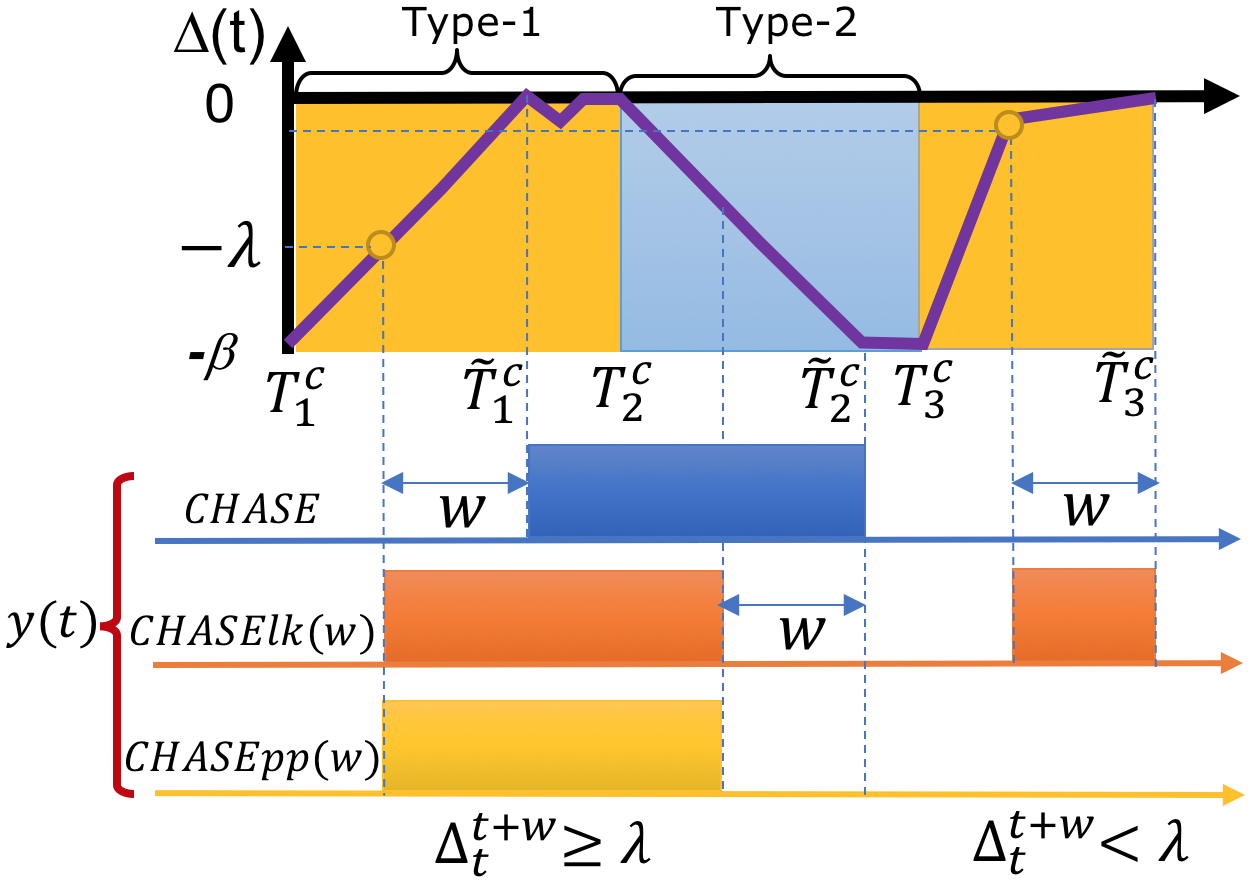}\caption{\label{fig:An-example-of-CHASE}An example of $\Delta(t)$ and the online algorithms \textsf{CHASE}, \textsf{CHASElk}, and \textsf{CHASEpp}. The prediction-aware online algorithms detect the segment type $w$ time slots before the prediction-oblivious \textsf{CHASE}.}
	\vspace{-4mm}
\end{figure} 
\subsection{Online Algorithm With Prediction}
\textsf{CHASE} can be extended straightforwardly to the setting with prediction, where at each time slot $t$ the precise prediction of the input for a window of $w$ time slots $\big[\sigma(\tau)\big]_{t}^{t+w}$, is available. As one can see from the Fig.~\ref{fig:An-example-of-CHASE}, the algorithm \textsf{CHASElk$(w)$}  \cite{Minghua2013SIG} behaves exactly the same as \textsf{CHASE}, except that it can detect the critical segment type and turn on/off the generator $w$ time slots  ahead of \textsf{CHASE}. Hence, \textsf{CHASElk$(w)$} achieves a better competitive ratio that satisfies
\begin{equation}
 {\sf CR}(\textsf{CHASElk}(w)) \leq   3-2f(\alpha, w)\leq 3-2\alpha , \label{eq:crchaselk}
\end{equation} 
where 
\begin{equation}
f(\alpha, w)=\alpha+   \frac{(1-\alpha)}{1+ \beta \frac{Lc_{\mathrm{o}}+c_{\mathrm{m}}/(1-\alpha)}{wc_{\mathrm{m}}(Lc_{\mathrm{o}}+c_{\mathrm{m}})}}   
\end{equation}

As discussed in Sec.~\ref{sec:relwork}, \textsf{CHASElk$(w)$} achieves the best competitive ratio with prediction prior to our study.

Next, we tackle this problem from a different perspective and propose a new threshold-based online algorithm that is substantially different from the existing algorithms. 
This prediction-aware online algorithm attains the best competitive ratio to date by exploiting a new design space.  

\section{Novel Prediction-Aware Online Algorithm}
\label{sec:CHASE-pp}
\subsection{Intuitions}
Consider the first and second type-1 critical segments in Fig.~\ref{fig:An-example-of-CHASE}. For both of these segments the previous algorithm \textsf{CHASElk($w$)} detects the segment type at $t=\tilde{T_{1}^{c}}-w$ and $t=\tilde{T_{3}^{c}}-w$, respectively and turns on the generator.
Using these two examples, we explain two intuitions that motivate us to design better prediction-aware online algorithms.
\begin{itemize}
\item The first intuition is that in the first type-1 segment the cumulative cost difference in the window $[\tilde{T_{1}^{c}}-w, \tilde{T_{1}^{c}}]$ is large ($\Delta(\tilde{T_{1}^{c}})-\Delta(\tilde{T_{1}^{c}}-w)=\lambda $). This means there is a substantial demand in the look-ahead window, and it is cost-effective to turn on the generator because spending the startup cost is worthy when we can enjoy the significant benefit of using the generator. Meanwhile, in the second type-1 segment, this cumulative cost difference is almost zero, which means sporadic demand in the look-ahead window. Hence, it is better to keep the generator off and avoid spending the high startup cost.  
\item The second intuition is that in the second type-1 segment, when $\Delta(\tilde{T_{3}^{c}})$ reaches zero $\Delta(\tilde{T_{3}^{c}}-w)$ is almost zero and it means we already suffered from not turning on the generator earlier. Hence, turning on the generator at the current time $\tilde{T_{3}^{c}}-w$, when there is not enough demand in the look-ahead window, is not beneficial. On the other hand, in the first type-1 segment at time $\tilde{T_{1}^{c}}-w$, we have $\Delta(\tilde{T_{1}^{c}}-w)= -\lambda$, which means that we are still at the beginning of the type-1 segment and turning on the generator at this moment is worthy. 
\end{itemize}

Following these intuitions, in the online setting, we turn on the generator only if we detect entering a type-1 critical segment, and there is a substantial benefit in using the generator in the look-ahead window. Meanwhile, as soon as we detect the type-2 critical segment, we should turn off the generator. Otherwise, the online algorithm will spend $c_{\mathrm{m}}$ unit per time slot of idling cost, while the offline has already turned off the generator and does not cost a penny.  This is similar to a ski-rental problem where keeping the generator on is like the skier keeps renting the ski, and its online cost keeps increasing while the offline algorithm bought the ski, and its cost is fixed. To capture the benefit of using the generator, we define an important parameter named cumulative differential cost in the prediction window.    
\begin{defn}    
For any $\tau \in [t,t+w]$, we define   $\Delta_{t}^{\tau}$ as the {\em cumulative differential cost} between using or not using the generator in the time interval $[t,\tau]$ as follows: \begin{equation}
 \Delta_{t}^{\tau} \triangleq  \sum_{s=t}^{\tau} \delta(s).
\end{equation}
\end{defn} 
This parameter utilizes all the predicted information, and it is critical for our novel online algorithm design. 

\subsection{Algorithm Description}
We denote our prediction-aware online algorithm as \textsf{CHASEpp($w$)}, as presented in Algorithm~\ref{alg:CHASE-pr}. In line 3 of the algorithm, if we detect being in a type-2 critical segment, we turn off the generator. Meanwhile, if we detect being in a type-1 critical segment (lines 5 to 13), we check for detecting the next type-2 critical segment. If we can not detect it by the end of the window (lines 7), we turn on the generator if we have $\Delta_{t}^{t+w} \geq \lambda $, which means there is enough benefit in using the generator for the prediction window. Similarly, if we can detect the next type-2 critical segment, and $\Delta_{t}^{\tau_2} \geq 0$ (line 9), we turn on the generator since there is enough benefit that can compensate the startup cost. Otherwise, we just keep the generator status unchanged.

\begin{algorithm}[htb!]   
{\caption{ ${\sf CHASEpp}(w) [t,\sigma(\tau)_{\tau = t}^{t+w}, y(t-1)]$} \label{alg:CHASE-pr}     
\begin{algorithmic}[1]
\STATE find $ \Delta(\tau)_{\tau = t}^{t+w}$  
\STATE set $\tau_1 \leftarrow \min\big\{\tau =t, ...,t+w \mid \Delta(\tau) = 0 \mbox{\ or\ }-\beta  \big\}$
\IF{ $\Delta(\tau_1)=-\beta$ (type-2)} 
\STATE {$y(t) \leftarrow 0$}
\ELSIF{$\Delta(\tau_1)=0$ (type-1)}
\STATE set $\tau_2 \leftarrow \min\big\{\tau =t, ...,t+w \mid \Delta(\tau)=-\beta \big\}$
\IF {$\Delta(\tau) > -\beta, \forall \tau \in [t,t+w],   \mbox{\ and \ } \Delta_{t}^{t+w} \geq \lambda$}
\STATE {$y(t) \leftarrow 1$} 
\ELSIF {$\Delta_{t}^{\tau_2} \geq 0$}  
\STATE {$y(t) \leftarrow 1$} 
\ELSE  
\STATE {$y(t) \leftarrow y(t-1)$}
\ENDIF
\ELSE  
\STATE {$y(t) \leftarrow y(t-1)$}
\ENDIF
\STATE set $u(t)$, $v(t)$, and $s(t)$ according to~\eqref{lem:fMCMP}
\end{algorithmic}
}
\end{algorithm}

We note that the online algorithm design space explored in this paper is new, in which turning on the generator depends on satisfying two conditions at the same time. First, we need to make sure the offline algorithm has turned on the generator. Second, we turn on the generator only if there is a significant benefit in using it for the future window. By comparing \textsf{CHASEpp($w$)} with the existing algorithms, one can see that the state-of-the-art algorithm \textsf{CHASElk($w$)} is a simple extension of \textsf{CHASE}, which tracks the offline solution in an online manner. But with checking the cumulative differential cost in the prediction window, \textsf{CHASEpp($w$)} makes smarter decisions when tracking the offline optimal to ensure better competitiveness. This new and effective design space improves the state-of-the-art competitive ratio.    
\section{Performance Analysis}
\label{sec:Performance_analysis}
\subsection{Competitive Ratio}
\label{subsec:worst}
Let us denote the algorithm with $\lambda \geq 0$ as its threshold to be $\textsf{CHASEpp}(w,\lambda)$. We have:
\begin{thm} The competitive ratio of Algorithm~\ref{alg:CHASE-pr} with $0 \leq \lambda \leq \beta$ is   
\label{lem:crfunction}  
\begin{equation}
{\sf CR}(\textsf{CHASEpp}(w,\lambda)) = \max \{R_{\mathrm{on}}(\lambda), R_{\mathrm{off}}(\lambda)\}, 
\end{equation}  
where $R_{\mathrm{on}}(\lambda)$ is a decreasing function given by  
\begin{eqnarray} 
\label{eq:RON}
&& R_{\mathrm{on}}(\lambda) =  1+ \big(1-  \frac{Lc_{\mathrm{o}}+c_{\mathrm{m}}}{L(p_{\mathrm{max}}+\eta\cdot c_{\mathrm{g}})} \big) \cdot \\
&&  \max\limits_{{q \in \{0,wc_{\mathrm{m}}\}}} \big\{ \frac{ 2\beta-q}{  \beta+\big(2wc_{\mathrm{m}}-q+\frac{c_{\mathrm{o}}}{p_{\mathrm{max}}+\eta\cdot c_{\mathrm{g}}}\lambda \big) \big(1-\frac{c_{\mathrm{m}}}{L(p_{\mathrm{max}}+\eta\cdot c_{\mathrm{g}}-c_{\mathrm{o}})} \big) } \big\}, \notag
\end{eqnarray}    
 and $R_{\mathrm{off}}(\lambda)$ is an increasing function given by
\begin{equation}   
\label{eq:ROFF}
 R_{\mathrm{off}}(\lambda)=\frac{wc_{\mathrm{m}}+\lambda}{wc_{\mathrm{m}}+\frac{c_{\mathrm{o}}}{p_{\mathrm{max}}+\eta\cdot c_{\mathrm{g}}}\lambda}.   
\end{equation}   
 \end{thm} 
 
 All the results presented in this paper have rigorous proofs, but due to page limit, we only provide a sketch of the idea behind our proof and present the details in Appendix~\ref{sec:appendix A}.

\begin{proof}[Sketch of Proof] One can show that this algorithm has two and only two possible worst-case inputs. We explain each of these inputs, and the competitive ratio can be found by finding the maximum of the performance ratio among these two values.    
   
  \begin{figure}  
\begin{minipage}[b][1\totalheight][t]{0.48\columnwidth}%
\begin{center} 
\includegraphics[width=1\columnwidth]{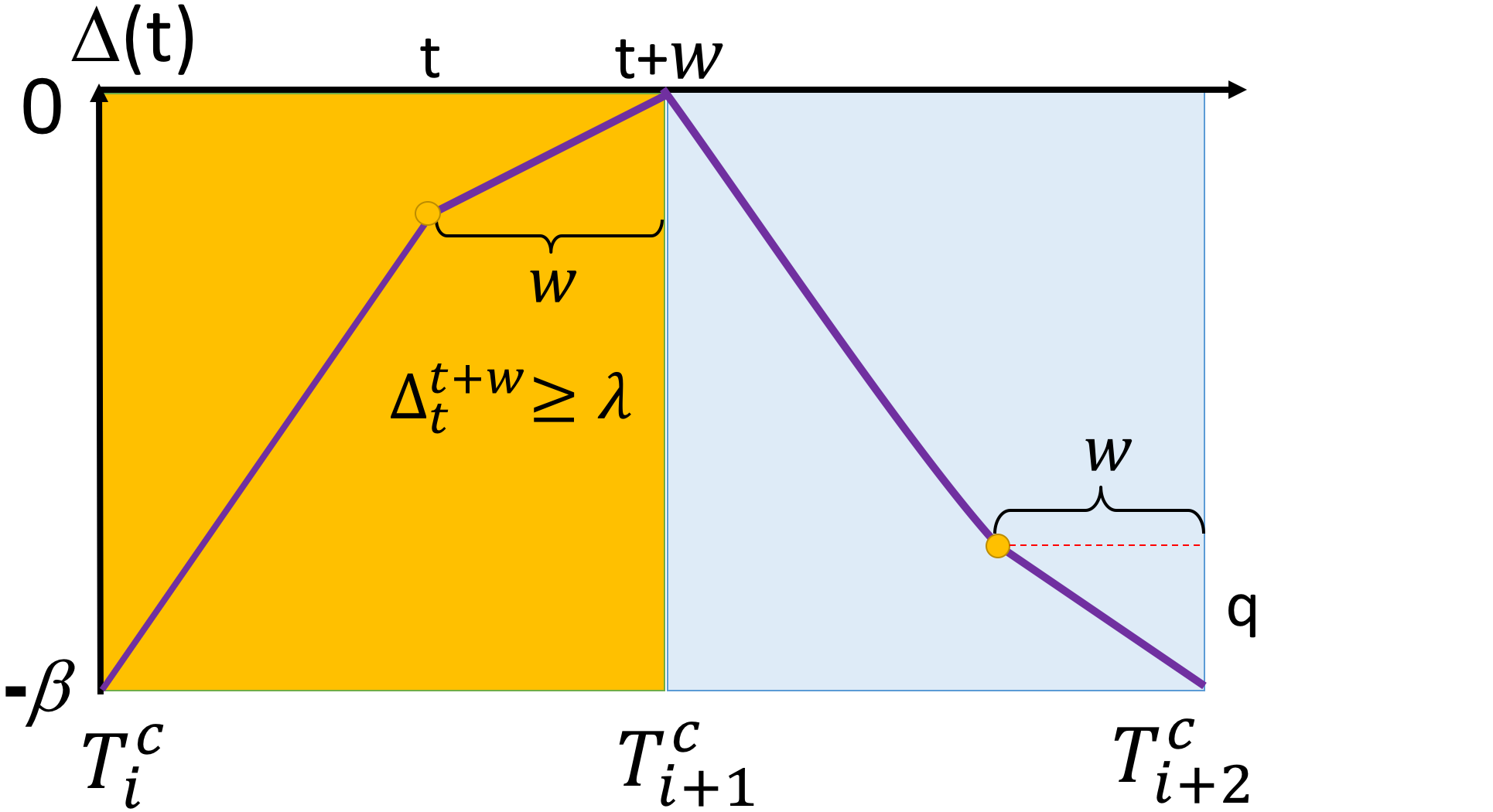} 
\par\end{center}   
\caption{\label{fig:Ron} First worst-case input for the type-1 and type-2 critical segments.}
\end{minipage}\hfill{}%
\begin{minipage}[b][1\totalheight][t]{0.48\columnwidth}%
\begin{center} 
\includegraphics[width=1\columnwidth]{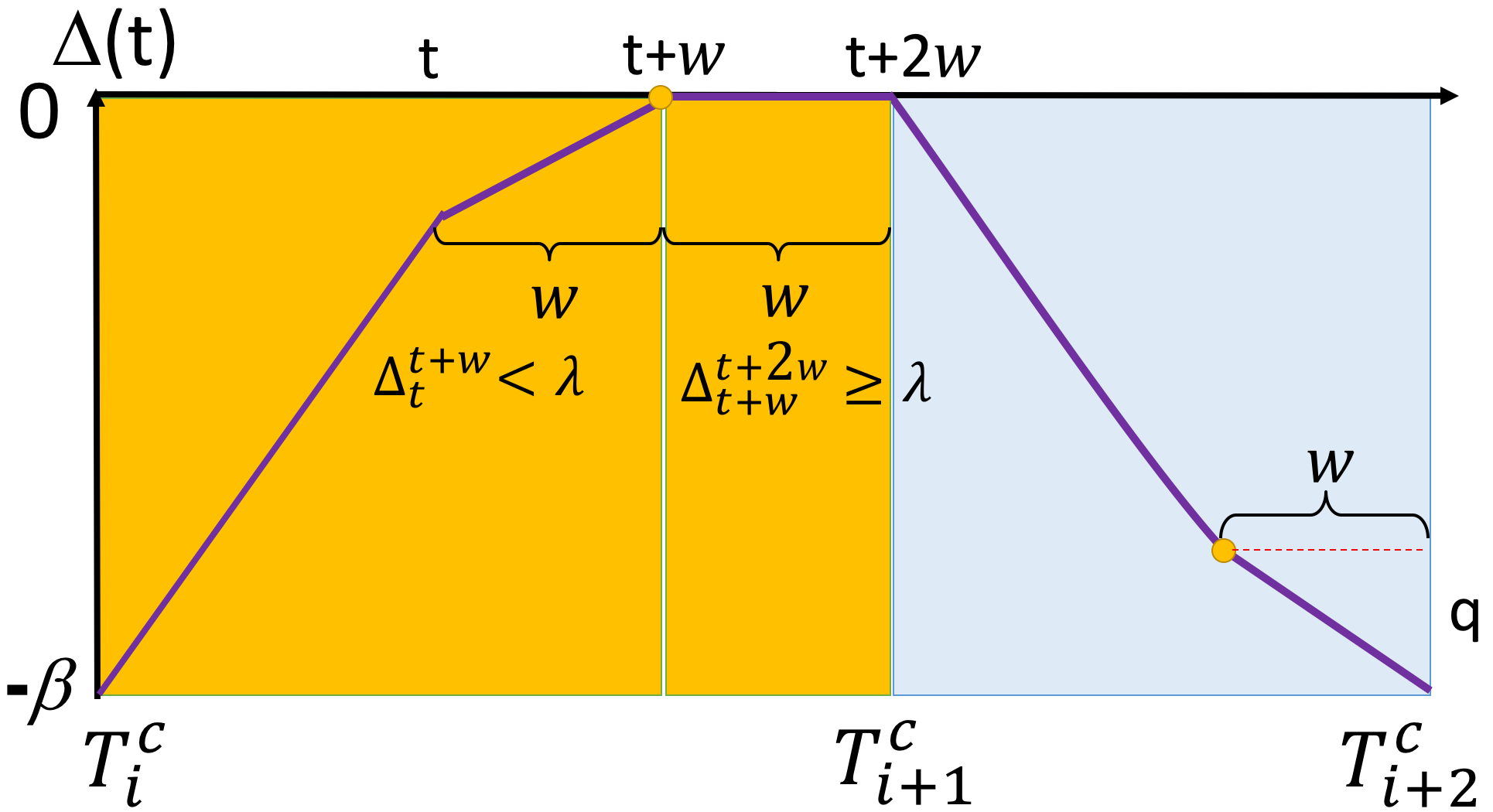}
\par\end{center} 
\caption{\label{fig:Roff} Second worst-case input appears by keeping the generator off for a window.} 
\end{minipage}    
	\vspace{-4mm}
\end{figure}

First worst-case input: Consider the input in Fig.~\ref{fig:Ron}. 
The algorithm turns on the generator at  $t=T_{i+1}^{c}-w $ and turns it off at $t=T_{i+2}^{c}-w $. We calculate performance ratio for this input as $R_{\mathrm{on}}(\lambda)$ presented in~\eqref{eq:RON}, where $q= \Delta(T_{i+2}^{c}-w)- \Delta(T_{i+2}^{c})$. This performance ratio is calculated by finding the maximum across the possible values of $q$. 

Second worst-case input: Consider the example in Fig.~\ref{fig:Roff}. At time $t$ we have $\Delta_{t}^{t+w}<\lambda$ and the algorithm keeps the generator off, but at time $t+w$ we have $\Delta_{t+w}^{t+2w}\geq \lambda$ and the algorithm turns on the generator. When calculating the performance ratio for this input, we observe that the ratio of the online cost increment over the offline cost increment is upper bounded by $R_{\mathrm{off}}(\lambda)$ presented in~\eqref{eq:ROFF}. Hence, the competitive ratio is equal to the maximum of these two values. 

\end{proof}  
\subsection{The Optimal Threshold}   
We find the optimal threshold $\lambda^*$ that minimizes the competitive ratio ${\sf CR}$
\begin{equation}
\lambda^*= \arg\min\limits_{{ \lambda }}\max \{R_{\mathrm{on}}(\lambda), R_{\mathrm{off}}(\lambda)\}.
\end{equation}
To understand how to find $\lambda^*$ and its corresponding ${\sf CR}$, we consider the example given in Fig.~\ref{fig:findaplot}. In this example we observe that at $\lambda=0$,  we have $R_{\mathrm{off}}(0) \leq R_{\mathrm{on}}(0) $. Meanwhile, from Theorem~\ref{lem:crfunction} we know that $R_{\mathrm{on}}(\lambda)$ is always a decreasing function and $R_{\mathrm{off}}(\lambda)$ is always an increasing function. Hence, $\lambda^*$ can be computed by finding the intersection of these two functions: 
\begin{subequations} 
\label{eq:finda} 
	\begin{eqnarray}
	 \lambda^*=	& \max & \lambda \\
		&\mbox{s.t.}& 0 \leq \lambda \leq \beta,\label{C_a_beta}\\
		&& 0 \leq \lambda \leq L \big(p_{\mathrm{max}}+\eta\cdot c_{\mathrm{g}}-c_{\mathrm{o}}- \frac{c_{\mathrm{m}}}{L}\big)w,\label{C_a_window}\\  
		&& R_{\mathrm{on}}(\lambda) \geq  R_{\mathrm{off}}(\lambda),\label{C_a_optimal}
	\end{eqnarray} 
\end{subequations} 
where~\eqref{C_a_beta} ensures the threshold is not larger than the startup cost; otherwise, it is always optimal to turn on the generator. The constraint~\eqref{C_a_window} ensures that the threshold is within the maximum possible value, and ~\eqref{C_a_optimal} ensures that we find the threshold that gives us the minimum competitive ratio. 
We can find the intersection of the two functions by using a simple binary search. As we can see in Fig.~\ref{fig:threshold}, by increasing the window size, the value of the threshold over the startup cost ($\lambda^*/\beta$) monotonically increases and approaches 1.      

Now that the optimal threshold is calculated, we present the competitive ratio of our proposed online algorithm:
\begin{thm}
\label{ref: competitive ratio}
 The competitive ratio of the algorithm \textsf{CHASEpp$(w)$} satisfies 
\begin{eqnarray}
 {\sf CR} \leq  3-2g(\alpha, w), \label{eq:competitive-ratio}
\end{eqnarray}
where     
\begin{align} 
\label{eq:g-function}
&g(\alpha, w) = \alpha + (1-\alpha) \Big( 1 - \tfrac{1}{2} \cdot \max\limits_{{q \in \{0,wc_{\mathrm{m}}\}}} \big\{ \\
&\frac{ (2\beta-q)}{  \beta+\big((2wc_{\mathrm{m}}-q)(Lc_{\mathrm{o}}+c_{\mathrm{m}})+ \alpha L c_{\mathrm{o}}   \lambda^*\big) /\big(Lc_{\mathrm{o}}+c_{\mathrm{m}}/(1-\alpha)\big) } \big  \} \Big)    \notag 
\end{align}     
which is $\alpha +(1-\alpha)/(1+\mathcal{O}(\frac{1}{w})) $ and monotonically increases from $\alpha$ to $1$ as $w$ increases. 
\end{thm} 
 \begin{figure}   
\begin{minipage}[b][1\totalheight][t]{0.48\columnwidth}%
\begin{center} 
\includegraphics[width=1\columnwidth]{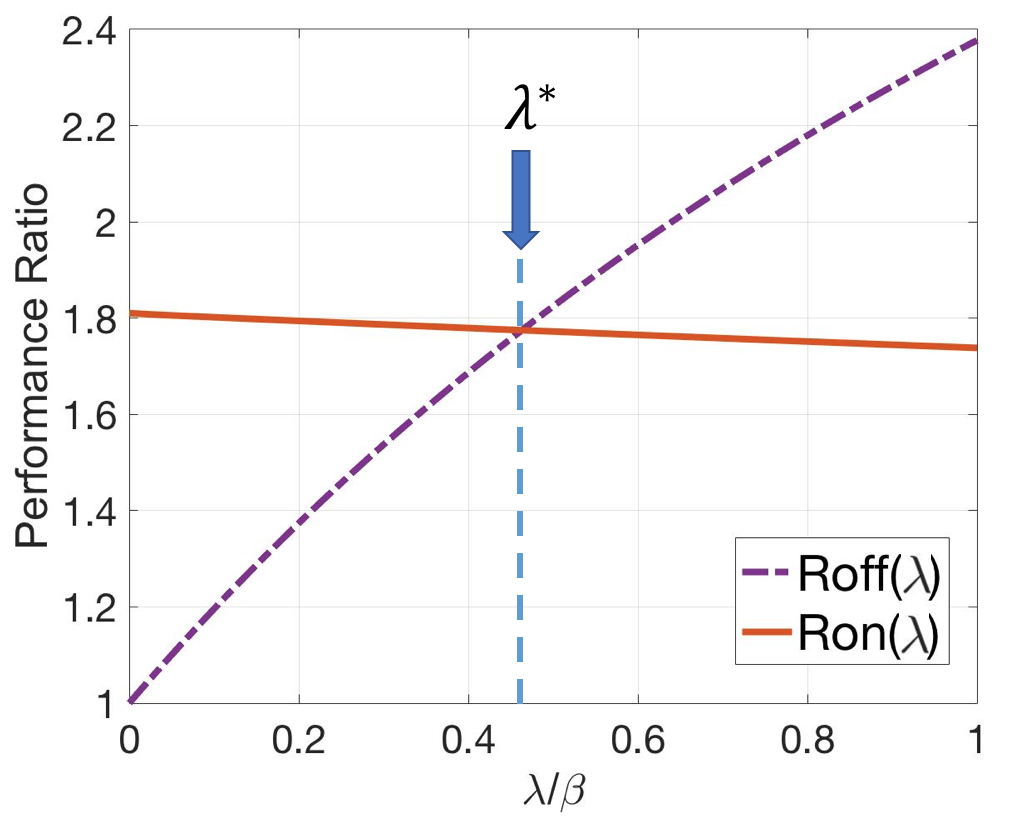} 
\par\end{center}
\caption{\label{fig:findaplot}The value of $\lambda^*$ is found at the intersection of the two functions $R_{\mathrm{on}}(\lambda)$ and $R_{\mathrm{off}}(\lambda)$.}    
\end{minipage}\hfill{}%
\begin{minipage}[b][1\totalheight][t]{0.48\columnwidth}%
\begin{center} 
\includegraphics[width=1\columnwidth]{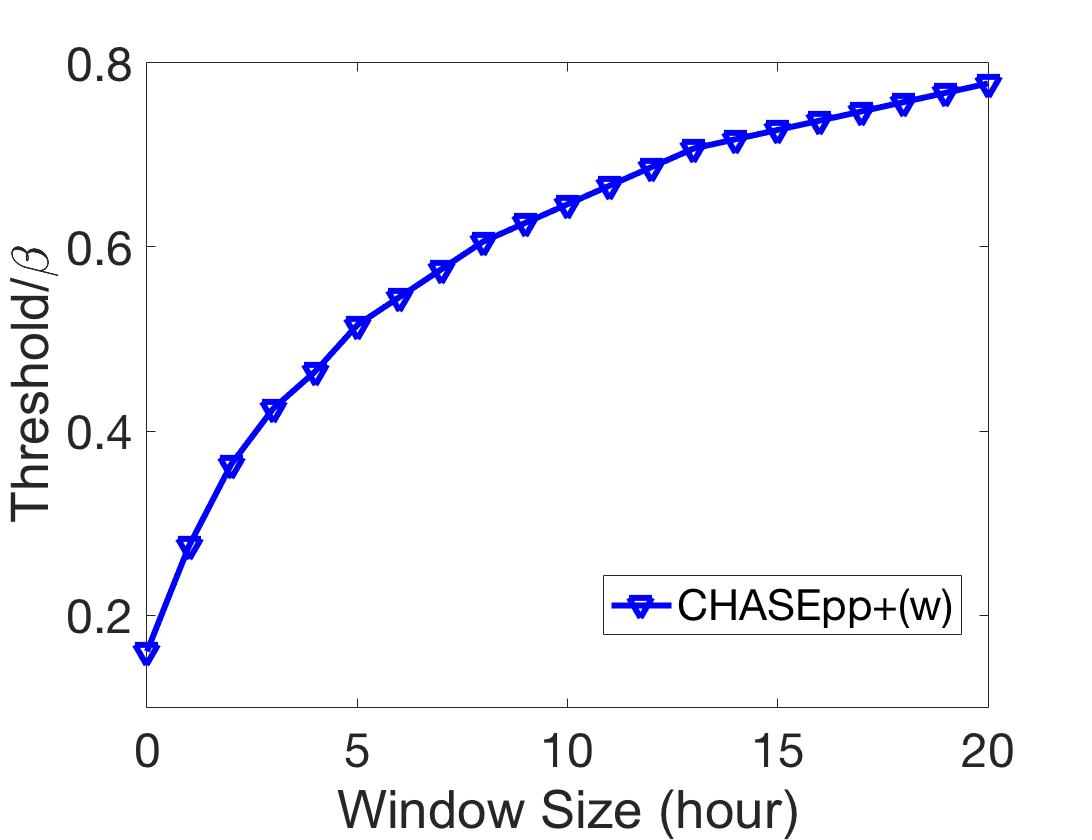}
\par\end{center} 
\caption{\label{fig:threshold}The value of the threshold over the startup cost $\beta$ for different window sizes ($w$).} 
\end{minipage} 
\vspace{-4mm}
\end{figure}
\begin{proof}
   Refer to  Appendix~\ref{sec:appendix C}.
\end{proof}
If there is no prediction ($w=0$), we have $g(\alpha,0)=\alpha$, and~\eqref{eq:competitive-ratio}  reduces to the competitive ratio of \textsf{CHASE}. Meanwhile, if $w$ is large,~\eqref{eq:g-function} turns to
 \begin{align}\label{eq:crfunction}
 & g(\alpha, w)=\alpha+  
 \\
&\frac{(1-\alpha)}{1+  \beta(Lc_{\mathrm{o}}+c_{\mathrm{m}}/(1-\alpha))/ \underbrace{(2wc_{\mathrm{m}}(Lc_{\mathrm{o}}+c_{\mathrm{m}})+\alpha Lc_{\mathrm{o}}\lambda^*)}_{(\dagger)} }.  \notag \end{align}  
One can see that by increasing the value of $\lambda^*$ or $w$, value of the function $g(\alpha, w)$ keeps increasing thus, the competitive ratio keeps decreasing.   
To understand how the new algorithm improves the competitive ratio of \textsf{CHASElk$(w)$}, we compare $g(\alpha, w)$ with $f(\alpha,w)$ presented in~\eqref{eq:crchaselk}. One can see that in $g(\alpha,w)$, the term denoted as $(\dagger)$ is larger than the corresponding part in  $f(\alpha,w)$ in~\eqref{eq:crchaselk} by $(wc_{\mathrm{m}}(Lc_{\mathrm{o}}+c_{\mathrm{m}})+\alpha Lc_{\mathrm{o}}\lambda^*)$. This means that the competitive ratio decreases twice faster with respect to $w$ than the previous algorithm. Therefore, our new algorithm always achieves better competitive ratio than the state-of-the-art algorithm \textsf{CHASElk$(w)$} as much as $20\%$ improvement as shown in Fig.~\ref{fig:alphaomegadif} to be discussed later.
 \begin{algorithm}[htb!]  
{\caption{ ${\sf CHASEpp}^+(w) [t,\sigma(\tau)_{\tau = t}^{t+w}, y(t-1)]$} \label{alg:CHASE-pr+}     
\begin{algorithmic}[1]
\IF{ $\tfrac{1}{\alpha}< {\sf CR}({\sf CHASEpp}(w))$ } 
\STATE {$y(t) \leftarrow 0, \, \, \, u(t) \leftarrow 0, \, \, \, v(t) \leftarrow a(t), \, \, \, s(t) \leftarrow h(t) $}
\STATE {return $(y(t), u(t), v(t), s(t))$}   
\ELSE  
\STATE {return ${\sf CHASEpp}(w) [t,\sigma(\tau)_{\tau = t}^{t+w}, y(t-1)]$}
\ENDIF 

\end{algorithmic}
} 
\end{algorithm}

Note that from the definition, maximum value of $R_{\mathrm{off}}(\lambda)$ is $1/\alpha$, which is the competitive ratio of the algorithm that only uses the external supply. Hence, when ${\sf CR}=  \max \{R_{\mathrm{on}}(\lambda^*), R_{\mathrm{off}}(\lambda^*)\}\geq 1/\alpha$, it means that instead of using our algorithm it is better to never turn on the generator. We summarize this result in Algorithm~\ref{alg:CHASE-pr+} denoted as ${\sf CHASEpp}^+(w)$.   \begin{cor}
The competitive ratio of ${\sf CHASEpp}^+(w)$ satisfies
\begin{equation}
     {\sf CR}({\sf CHASEpp}^+(w))= \min \{{\sf CR}({\sf CHASEpp}(w)), \tfrac{1}{\alpha}\}.  
\end{equation}
\end{cor} By the same logic ${\sf CHASElk}^+(w)$ can be defined. In the rest of this paper, we evaluate the performance of these new algorithms. In Fig.~\ref{fig:alphaomega}, this competitive ratio is depicted as a function of $\alpha$ and $w$.   
When $\alpha$ is large, it means the economic advantage of using local generation over external sources is small. Hence, both online and offline algorithms tend to use local generation less often. This improves the competitiveness of online algorithms. Thus, by increasing $\alpha$ from $0$ to $1$, the competitive ratio decreases from $3$ to $1$ monotonically.           

\section{Multiple Generator Scenario}
\label{sec:multigenerator}
In this section, we solve the multi-generator \textbf{MCMP} problem presented in Sec.~\ref{sec:pf}, where we have N units of heterogeneous generators. Without loss of generality, we assume that $L_1 \geq L_2 ... \geq L_{\mathrm{N}}$. The key is to slice the demand into multiple layers, assign each sub-demand to a different generator, and solve an optimization problem for a single generator. For partitioning, we start from the bottom up, and we slice the demand such that the $n$-th layer has at most $a^{\mathrm{ly-n}}=L_{\mathrm{n}}$ units of electricity demand and $h^{\mathrm{ly-n}}=\eta \cdot 
L_{\mathrm{n}}$ units of heat demand. Once we used all the local generation capacities,  for the remaining demands $(a^{\rm top}, h^{\rm top})$, we use the external sources. In this way, the bottom layers have the least frequent variations of demand, and we assign these layers to the generators with larger capacities. As a result, we use fewer generators, and these generators observe less variations, which helps to reduce the startup cost. 
The following theorem captures the performance of offline and online algorithms.  
\begin{thm} \label{thm:nOFA-optimal}
The offline algorithm that uses the layering approach produces an optimal offline solution for \textbf{MCMP}, and the online algorithm achieves the following competitive ratio: 
\begin{equation}
 {\sf CR} \leq  3-2g(\alpha_1, w), 
\end{equation}
where $\alpha_1= \frac{c_{\mathrm{o}}+c_{\mathrm{m}}/L_1}{p_{\mathrm{max}}+\eta c_{\mathrm{g}}}$, and $g(\alpha, w) $ is defined in Theorem~\ref{ref: competitive ratio} .
\end{thm}

\begin{proof}
   Refer to  Appendix~\ref{sec:appendix M}.
\end{proof}

\begin{figure}[t]  
	\centering{}\includegraphics[width=.48\columnwidth]{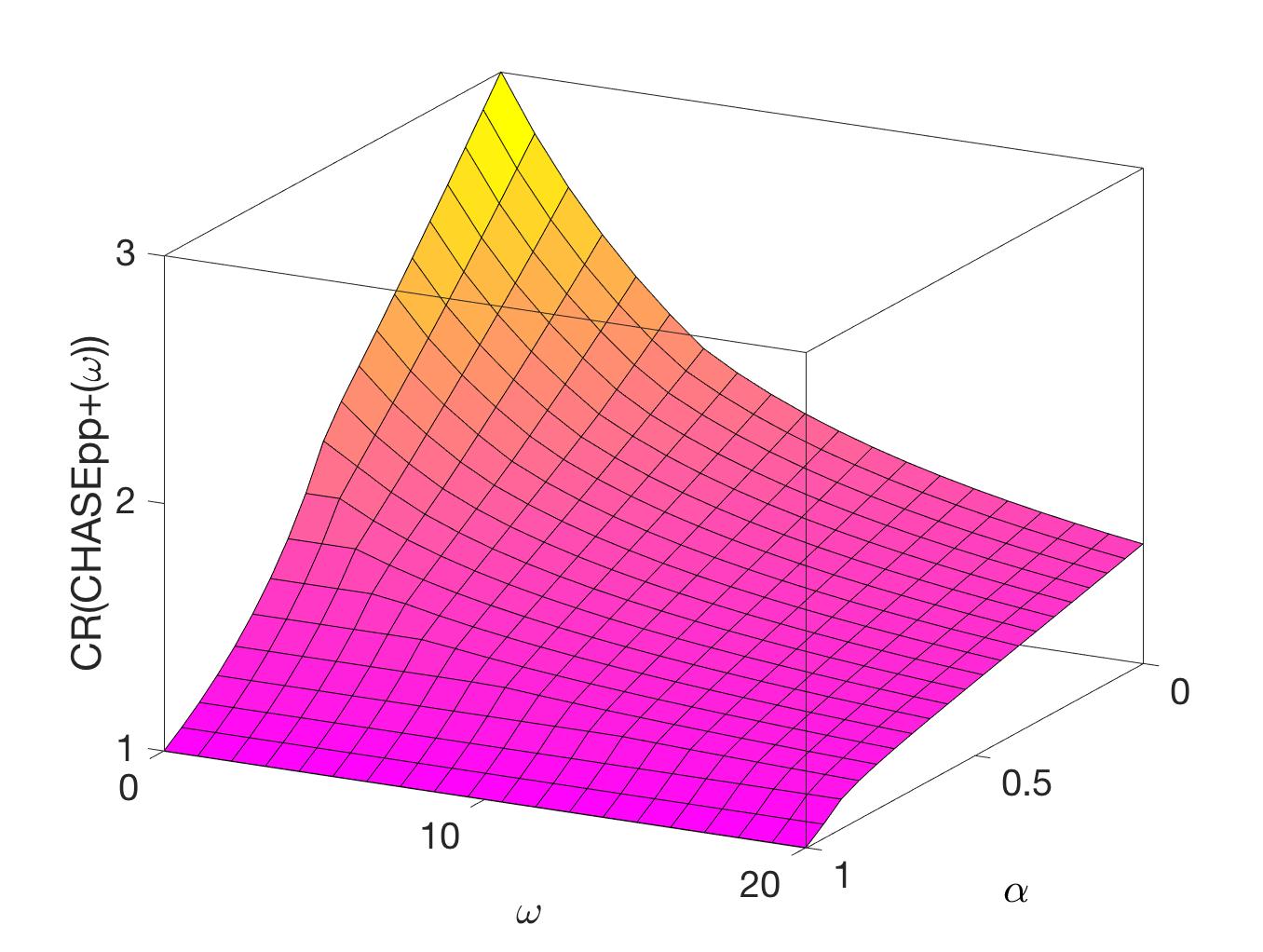}\caption{\label{fig:alphaomega}Competitive ratio of the  algorithm $\textsf{CHASEpp}^+(w)$ as a function of $\alpha$ and $w$.}  
	\vspace{-3mm}
\end{figure} 

In the heterogeneous setting, each generator has its own unique prediction-aware online algorithm with a different optimal threshold that depends on their generation capacity $L_{\mathrm{n}}$.
The significance of this result is that it extends the applicability of our online algorithm beyond the homogeneous setting.

\section{Lower bound of the Competitive Ratio}
\label{sec:lowerbound}

\begin{figure}   
\begin{minipage}[b][1\totalheight][t]{0.48\columnwidth}%
\begin{center}
\includegraphics[width=1\columnwidth]{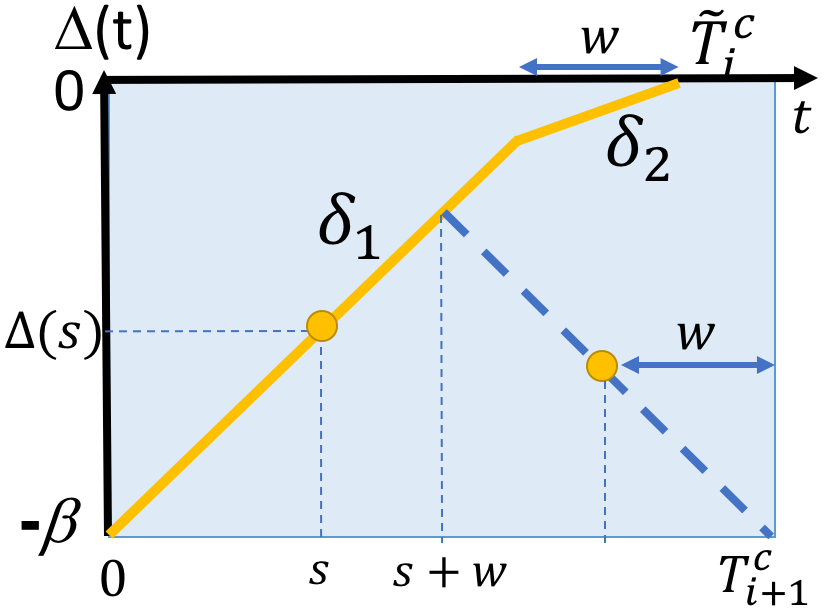}
\par\end{center}
\caption{\label{fig:lbexample} An example of the worst-case input with $(\delta_1,\delta_2)$.}
\end{minipage}\hfill{}%
\begin{minipage}[b][1\totalheight][t]{0.48\columnwidth}%
\begin{center} 
\includegraphics[width=1\columnwidth]{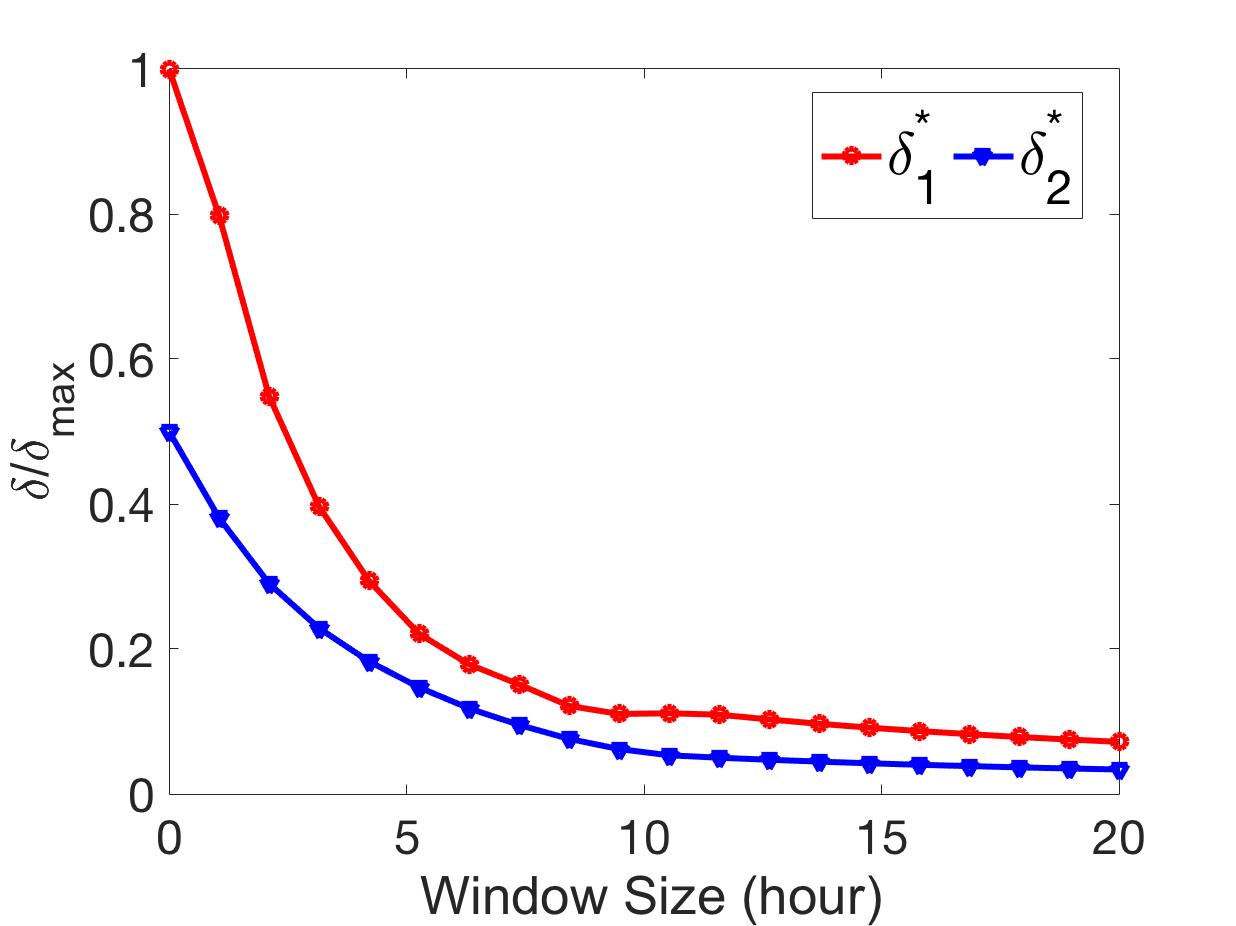}
\par\end{center} 
\caption{\label{fig:a1a2}Value of $(\delta_1^*,\delta_2^*)$, for different window sizes.} 
\end{minipage}
\vspace{-3mm}
\end{figure}

\begin{thm}\label{CRLOB} Let $\epsilon > 0 $ be the length of each time slot. As $\epsilon$ goes to zero and the discrete time setting approaches to the continues time setting, the competitive ratio for any prediction-aware deterministic online algorithm
${\mathcal A}$ for $\textsf{MCMP}_{\mathrm{s}}$ is lower bounded by 
\begin{equation} \label{bounddelta}
{\sf CR}({\mathcal A})\ge {\sf cr}(w)= \frac{c_{\mathrm{m}}+\delta_2^*}{c_{\mathrm{m}}+\frac{Lc_{\mathrm{o}} \alpha  }{Lc_{\mathrm{o}}+c_{\mathrm{m}}}\delta_2^*},
\end{equation}
where $\delta_2^*>0$ is an optimal objective of an optimization problem with system parameters as input.  
\end{thm}
\begin{proof}
   Refer to  Appendix~\ref{sec:appendix L}.
\end{proof}
\begin{proof}[Sketch of Proof] The key idea is that given any deterministic online algorithm $\mathcal A$, we progressively construct a particular worst-case input $\sigma(t) \triangleq (a(t), h(t), p(t))$, that for the performance ratio we have
\begin{equation}
  \frac{{\rm Cost}({\sf y}_{\mathcal A};\sigma)}{{\rm Cost}({\sf y}_{\rm OFA},\sigma)}     \geq {\sf cr}(w).
\end{equation}  
In what follows, we explain this input, and it's corresponding lower bound in Lemma~\ref{prlb}. Consider an example of the input in Fig.~\ref{fig:lbexample}. Starting from the beginning $t=0$, by giving a large value of differential cost ($\delta(t)=\delta_1$), the adversary tries to encourage the algorithm to turn on the generator. As soon as the algorithm turns on the generator at $s$, the adversary starts to hurt the algorithm as hard as possible by giving it $\delta(s+w)=-c_{\mathrm{m}}$. Hence, the $\Delta(t)$ function keeps decreasing until the algorithm turns off the generator at time $t=T_{i+1}^{c}-w $.   
We denote the performance ratio for this input as $PR(s)$. One can see that at the beginning, the adversary gives a larger differential cost  ($\delta(t)=\delta_1$) as the input. As time goes by and the offline cost changes, the adversary continues by giving a smaller differential cost $\delta_2$ as input.   
\begin{lem} \label{prlb} The lower bound of the competitive ratio is
\begin{eqnarray} 
{\sf CR}_{\mathrm{l}}(\delta_1,\delta_2) = \min \{R_{\mathrm{on}}(\delta_1,\delta_2), R_{\mathrm{off}}(\delta_2)\},
\end{eqnarray} 
where $R_{\mathrm{on}}(\delta_1,\delta_2)$ is given by  
\begin{eqnarray} 
 R_{\mathrm{on}}(\delta_1,\delta_2) =  \min_{s\in [1,(\beta-w\delta_2)/\delta_1]} PR(s),
\end{eqnarray}    
 and $R_{\mathrm{off}}(\delta_2)$ is given by
\begin{eqnarray} 
 R_{\mathrm{off}}(\delta_2)=\frac{c_{\mathrm{m}}+\delta_2}{c_{\mathrm{m}}+\frac{Lc_{\mathrm{o}} \alpha  }{Lc_{\mathrm{o}}+c_{\mathrm{m}}}\delta_2}.
\end{eqnarray} 
\end{lem} 

\begin{proof}
   Refer to  Appendix~\ref{sec:lemma1proof}.
\end{proof}

Similar to Theorem~\ref{lem:crfunction}, we find $(\delta_1^*,\delta_2^*)$ that maximizes the lower bound ${\sf CR}_{\mathrm{l}}(\delta_1,\delta_2)$. First, for each $\delta_2$, we find a corresponding $\delta_1$ such that $\delta_1= \text{arg}\max\limits_{ \delta }\, R_{\mathrm{on}}(\delta,\delta_2)$. This reduces $R_{\mathrm{on}}(\delta_1,\delta_2)$ to a single variable function of $\delta_2$. Now we find the maximum of the minimum of two single variable functions.  We know that for $\delta_2=0$, we have $R_{\mathrm{on}}(\delta_1,0) \geq R_{\mathrm{off}}(0)=1 $, and $R_{\mathrm{off}}(\delta_2) $ is an increasing function. Hence,  similar to \eqref{eq:finda} we keep increasing $\delta_2$ until we find the intersection of the two functions. Therefore, we always have $R_{\mathrm{on}}(\delta_1^*,\delta_2^*) \geq R_{\mathrm{off}}(\delta_2^*)$ and ${\sf CR}_{\mathrm{l}}(\delta_1^*,\delta_2^*) = R_{\mathrm{off}}(\delta_2^*)$, which completes the proof.
\end{proof} 
In Fig.~\ref{fig:a1a2}, we plot $(\delta_1^*,\delta_2^*)$ for different window sizes. As we can see, by increasing the window size, both $\delta_1^*$ and $\delta_2^*$ keep decreasing. As a result, the lower bound, which is a function of $\delta_2^*$, keeps decreasing.
\vspace{-2mm}   

\section{NUMERICAL EXPERIMENTS}
\label{sec:exp}
\begin{figure}  
\begin{minipage}[b][1\totalheight][t]{0.48\columnwidth}%
\begin{center} 
\includegraphics[width=1\columnwidth]{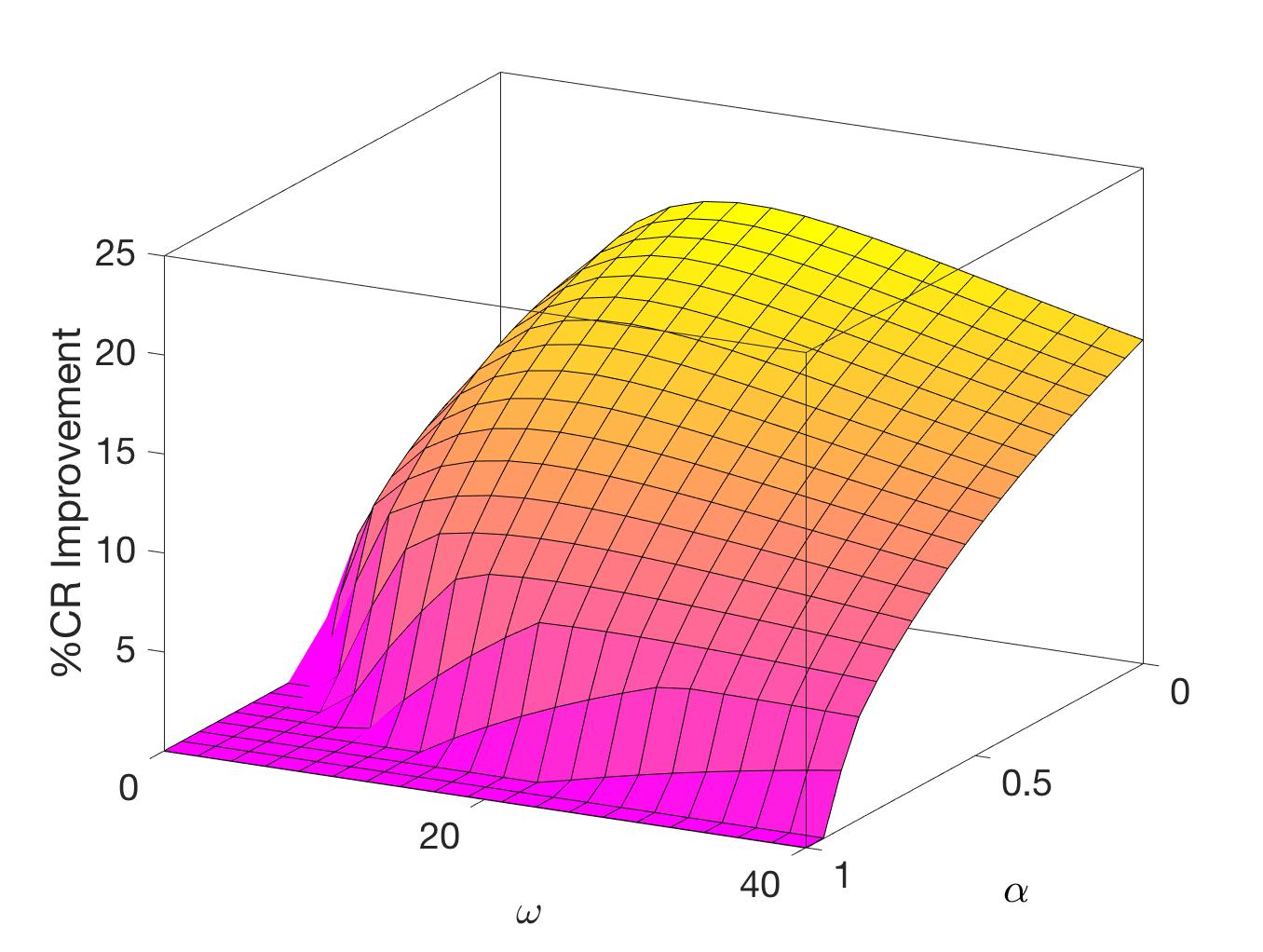}
\par\end{center} 
\caption{\label{fig:alphaomegadif} Competitive ratio improvement as a function of $\alpha$ and $w$.} 
\end{minipage}\hfill{}%
\begin{minipage}[b][1\totalheight][t]{0.48\columnwidth}%
\begin{center} 
\includegraphics[width=1\columnwidth]{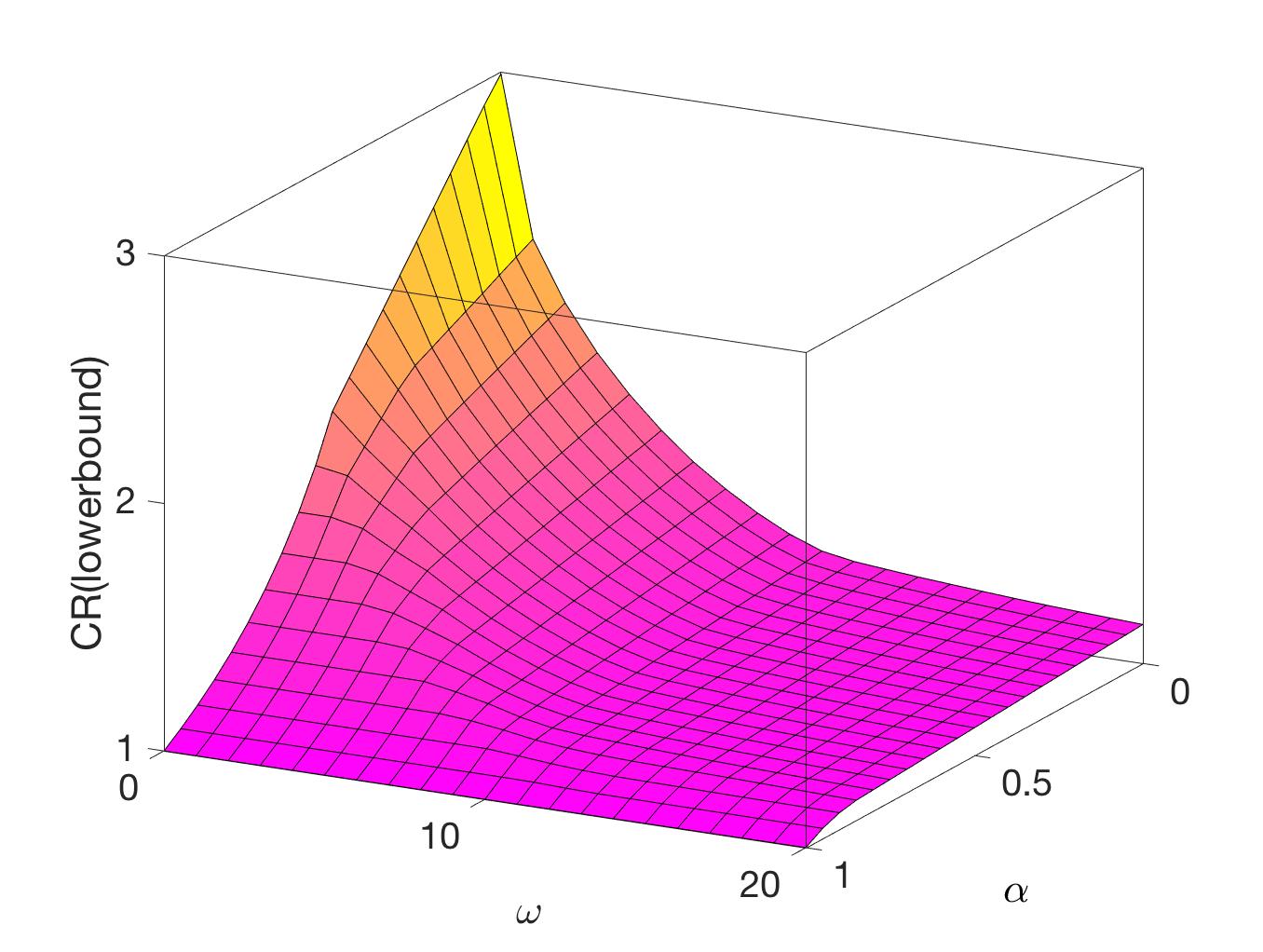}
\par\end{center}
\caption{\label{fig:crlb} Lower bound of the competitive ratio as a function of $\alpha$ and $w$.} 
\end{minipage}  
\vspace{-3mm}
\end{figure}  
We carry out numerical experiments using real-world traces to evaluate the performance of \textsf{CHASEpp}. We calculate the cost incurred by using only external electricity, heating, and wind energy (when no generator is utilized) as a benchmark, and we report the cost reduction of different algorithms compared to this benchmark. We compare performance of the optimal offline algorithm ${\sf OPT}$, ${\sf CHASE}$, ${\sf CHASElk}^+$, ${\sf CHASEpp}^+$, and ${\sf RHC}$, which is a popular algorithm widely used in the control literature \cite{Kwon1977modified}, with both perfect and noisy prediction. The length of each time slot is 1 hour and the total cost incurred during one week (T = 168) is reported. \vspace{-3mm}
\subsection{Experiment Setting} \label{expsetting}
We obtain the electricity and heat demand traces from \cite{CEUS}, which belongs to a college in San Francisco, with yearly electricity consumption of around $154GWh$, and gas consumption of around $5.1 \times 10^6$ therms. We use wind power traces from  \cite{NREL}, which are collected from a wind farm outside San Francisco with an installed capacity of $12MW$. 
We obtain the electricity and natural gas price from PG\&E \cite{PGE}, and we deploy generators with the same specifications as the one in \cite{tecogen}, with heat recovery efficiency $\eta$ set to be $1.8$. The incremental cost $c_{\mathrm{o}}$ and running cost $c_{\mathrm{m}}$ per unit time are set to be $\$0.051\ensuremath{/KWh}$ and $\$110/h$  respectively. We consider a heating system with the unit heat generation cost of $c_{\mathrm{g}}=\$0.0179/KWh$, according to \cite{greenenergy} and the startup cost  $\beta$ is set to be $\$1400$. The peak for the electricity demand is $30MW$, so we adopt $10$ generators with maximum capacity $1MW \times 3$, $3MW \times 4$, and $5MW \times 3$ to fully satisfy the demand. All the experiments are modeled and implemented in Matlab \cite{MATLAB:2021} using Gurobi optimization tools \cite{gurobi}.
\vspace{-1mm}
\subsection{Theoretical Ratio}

\begin{figure} 
\begin{minipage}[b][1\totalheight][t]{0.48\columnwidth}%
\begin{center}
\includegraphics[width=1.01\columnwidth]{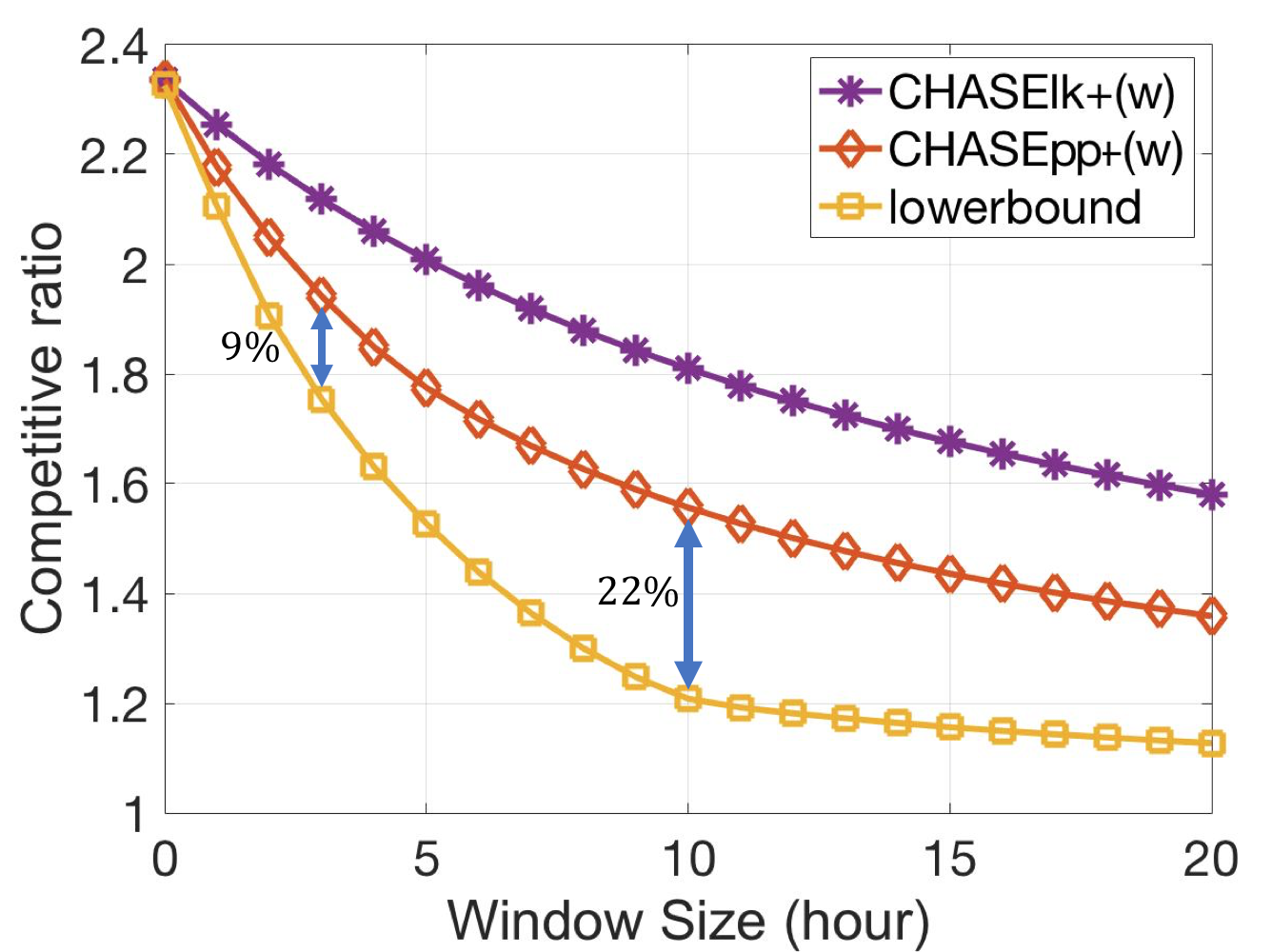}  
\par\end{center}    
\caption{\label{fig:comparecr} Competitive ratio as a function of prediction window size ($w$).}   
\end{minipage}\hfill{}%
\begin{minipage}[b][1\totalheight][t]{0.48\columnwidth}%
\begin{center}  
\includegraphics[width=1\columnwidth]{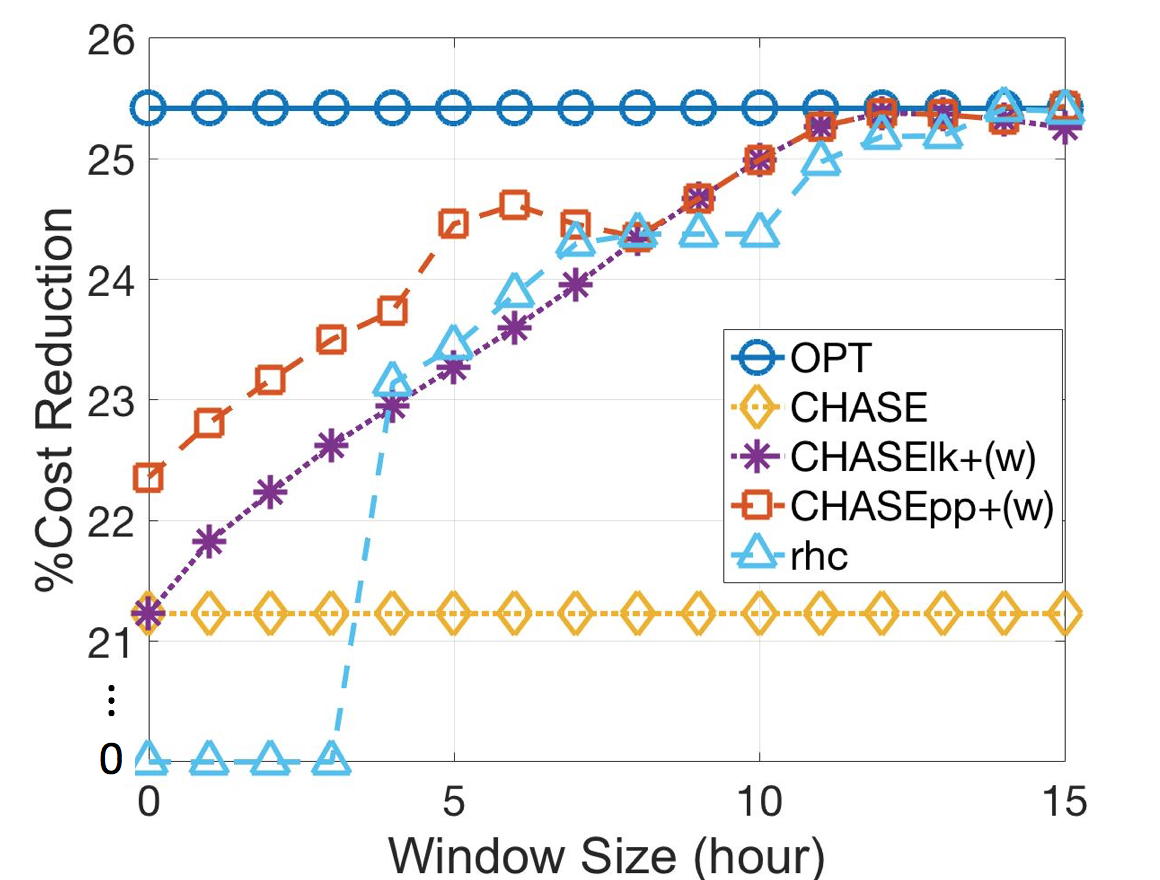}
\par\end{center}  
\caption{\label{fig:costreduction} Cost reduction as a function of prediction window size ($w$). }  
\end{minipage}     
\vspace{-3mm}
\end{figure}  
In Fig.~\ref{fig:alphaomegadif}, we plot the competitive ratio improvement of our algorithm $\textsf{CHASEpp}^+(w)$  over $\textsf{CHASElk}^+(w)$ as a function of $\alpha$ and $w$. We can see that our algorithm improves the competitive ratio by up to $20\%$. As expected by decreasing the value of $\alpha$, or increasing the window size, the competitive ratio improvement increases. But if we keep increasing the window size, both competitive ratios approach $1$, and the competitive ratio gap starts to decrease.
In a similar figure, the competitive ratio lower bound is depicted in Fig.~\ref{fig:crlb}. We can see that when $w$ approaches 0, the lower bound approaches $3-2\alpha$, which is the lower bound of the prediction-oblivious \textsf{CHASE}.  We also plot the competitive ratio and the lower bound for different values of $w$ in  Fig.~\ref{fig:comparecr}. Our algorithm's competitive ratio is always better than the previous algorithm, and it is not far from the lower bound. With a 3-hour prediction, our algorithm's competitive ratio is away from the lower bound by $9\%$. In the worst case with $w=10$ hours, our algorithm is away by at most $22\%$. One should note that in practice, 3 hours is a more typical prediction window size \cite{santhosh2020current}. 

\subsection{The Effect of Prediction Window}
In this section, we change the window size from $0$ to $15$, and we show the results in Fig.~\ref{fig:costreduction}. We observe that when the window size is large, all the algorithms perform very well and approach the optimal offline solution. On the other hand, when the window size is small, \textsf{RHC} performs very poorly while our online algorithm $\textsf{CHASEpp}^+(w)$ performs better than the previous algorithm. It is important to note that depending on the input structure, there may be a large performance discrepancy between $\textsf{CHASElk}^+(w)$ and $\textsf{CHASEpp}^+(w)$. In the following section, we show how our new online algorithm can improve the previous algorithm's performance by exploiting the structure of the predicted information.

{\bf Effect of the cumulative differential cost}: In this section, we build two inputs and depict the cost reduction of the two algorithms for these inputs in Fig.~\ref{fig:crdifference1} and Fig.~\ref{fig:crdifference2}. We observe that in the first input shown in Fig.~\ref{fig:crdifference1} with $a(t)=L$ when $\Delta(t+w)$ reaches zero the cumulative differential cost in the window is large, and hence both algorithms turn on the generator, and they have the same performance. On the other hand, in the second input shown in Fig.~\ref{fig:crdifference2} with $a(t)=L/4$ when $\Delta(t+w)$ reaches zero the cumulative differential cost is small, and the new algorithm does not turn on the generator and will not spend the additional startup cost, which leads to its better performance. 
Therefore, for a general input, if we have a large value of demand $a(t)=L$ coming in every time slot, both algorithms perform very well. But if the demand is small $a(t)=L/4$, the cumulative differential cost is also small, and using our new algorithm significantly improves the performance. 

\begin{figure}   
\begin{minipage}[b][1\totalheight][t]{0.48\columnwidth}%
\begin{center}
\includegraphics[width=.98\columnwidth]{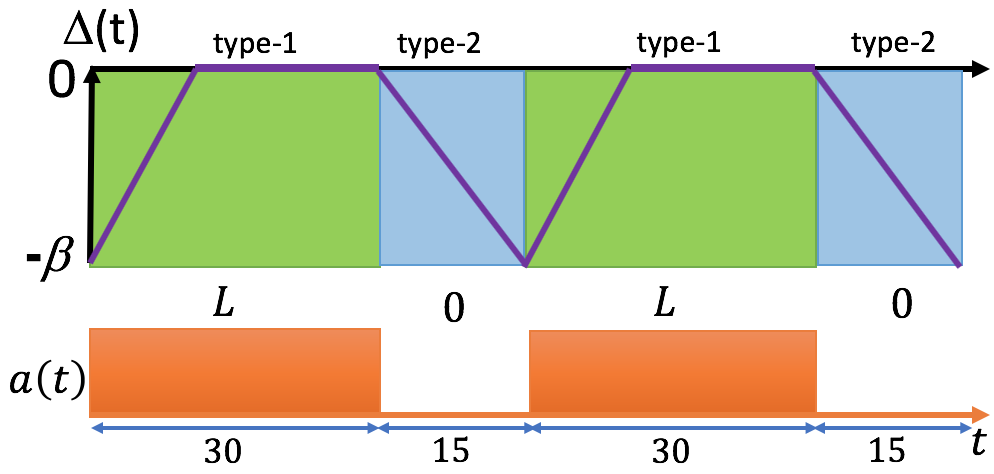}
\par\end{center}     
\begin{center} 
\includegraphics[width=1\columnwidth]{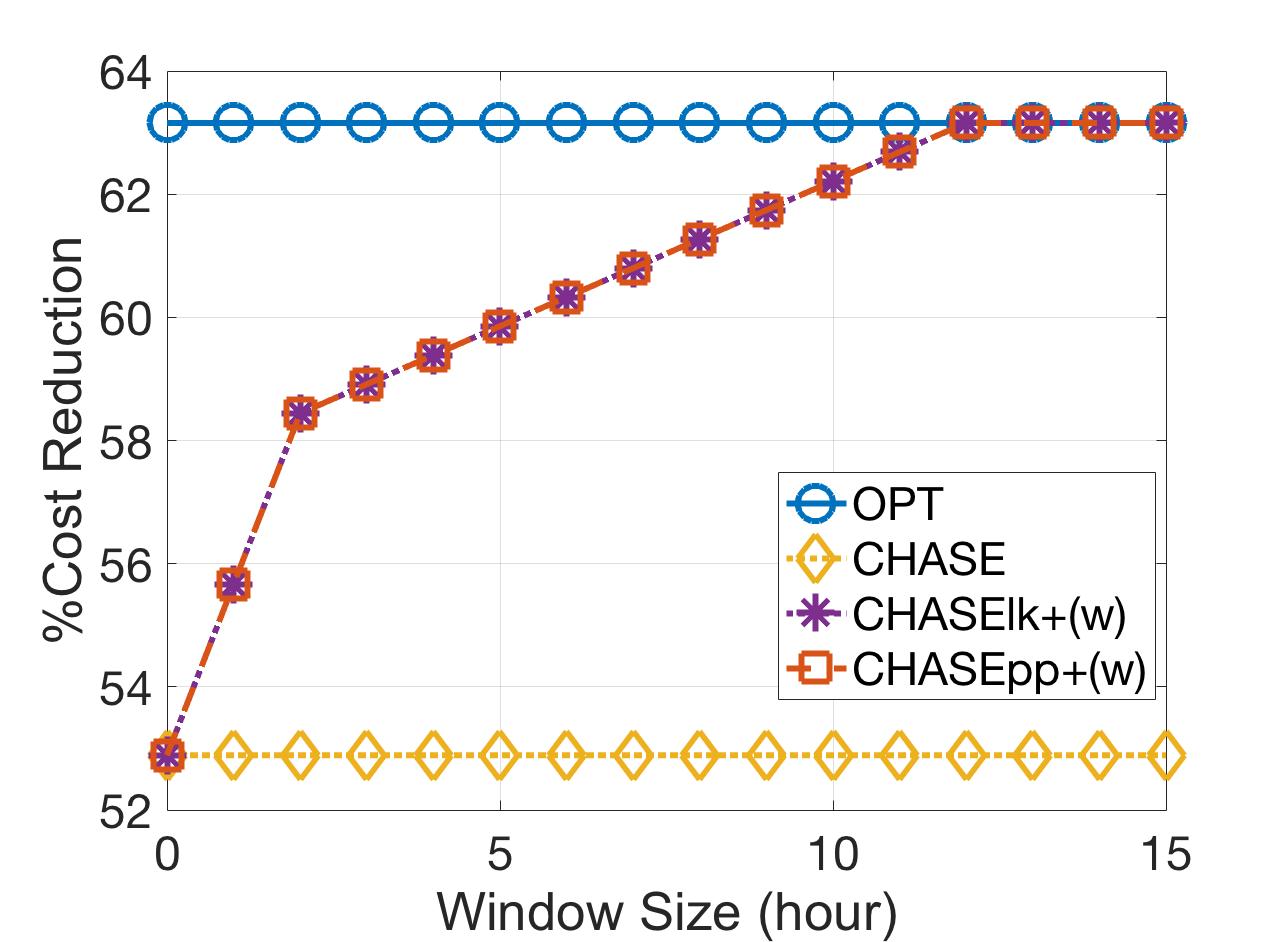}
\par\end{center}
\caption{\label{fig:crdifference1}TInput with $a(t)\in\{0,L\}$, $h(t)=\eta a(t)$, and $p(t)=p_{\mathrm{max}}$, where both algorithms perform the same.}  
\end{minipage}\hfill{}%
\begin{minipage}[b][1\totalheight][t]{0.48\columnwidth}%
\begin{center} 
\includegraphics[width=.98\columnwidth]{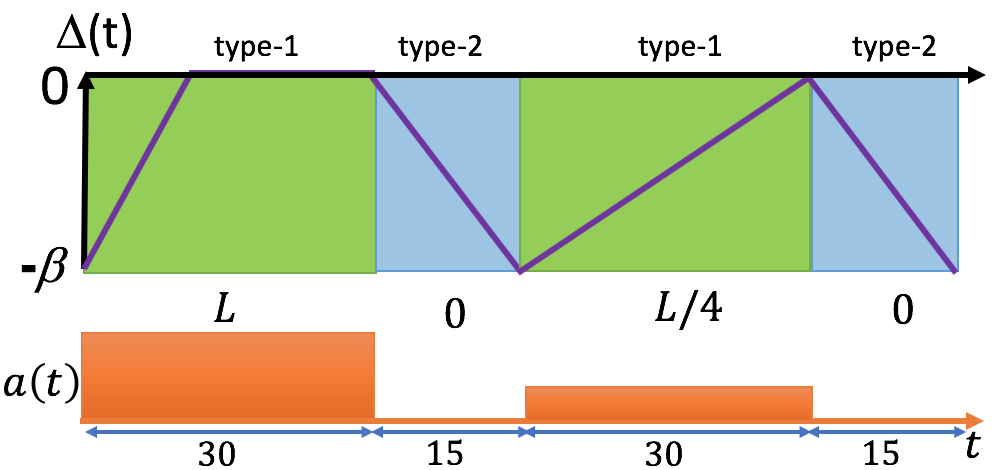} 
\par\end{center}       
\begin{center} 
\includegraphics[width=1\columnwidth]{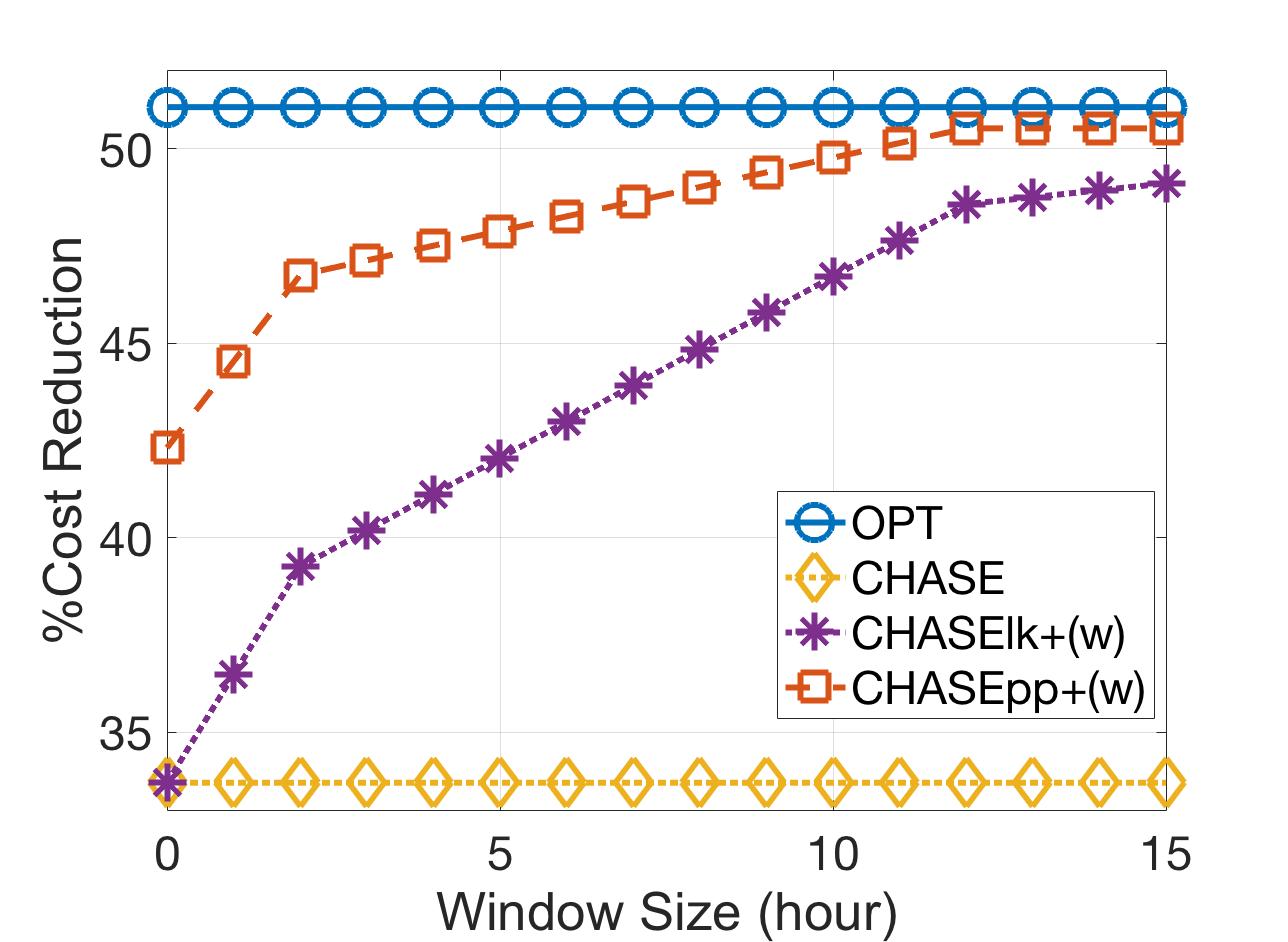}
\par\end{center}   
\caption{\label{fig:crdifference2}Input with $a(t)\in\{0,L/4,L\}$, $h(t)=\eta a(t)$, and $p(t)=p_{\mathrm{max}}$, where the new algorithm performs better.}     
\end{minipage}   
\vspace{-4mm}
\end{figure}

\subsection{The Effect of Prediction Error} \label{predictionerror}
While the day-ahead electricity demand prediction has an error range of $2-3\%$, the highly fluctuating nature of the wind power makes the next hour's prediction error to usually be around $20-40\%$ \cite{Wang2018MultistepAW}. Therefore, it is important to see how the prediction error can affect our online algorithm's performance. We obtain real-world wind power forecasting error distributions from \cite{hodge2012wind}, where the mean and standard deviation of the errors are based on the typical forecasts in the U.S., and hyperbolic distribution is used to represent the error. In  \cite{hodge2012characterizing} it has been shown that hyperbolic distribution is superior to the normal distribution in capturing wind power forecasting error. Still, to compare these two, we generate wind power forecasting errors from both. We start with the real-world hyperbolic distribution and zero-mean Gaussian, and in each time slot, we add the errors to the actual values. We also increase the standard deviation by $0$ to $100\%$ of the total installed capacity and the total peak demand for the wind power error and the heat demand error, respectively. In Fig.~\ref{fig:crerror1} and~\ref{fig:crerror3}, we report the average cost reduction of the algorithms over $100$ runs for two different lookahead window sizes of $1$ and $3$ hours with Gaussian errors. In Fig.~\ref{fig:hype1} and~\ref{fig:hype3} the results of the simulation for the real-world prediction errors with hyperbolic distribution are shown. It is important to note that for a $3$-hour prediction window size the errors are often in $20-40\%$ range \cite{Wang2018MultistepAW}. Therefore, by increasing the standard deviation up to $100\%$, we are stress-testing the algorithm. When $w$ is small, both algorithms are robust to the prediction error. When the window size increases, however, the previous algorithm becomes more sensitive, and its performance starts deteriorating. Because, for a large window size, the prediction error of each time slot aggregates over the window, and if the window size becomes too large, the prediction can even worsen the algorithm's performance. On the other hand, the new online algorithm keeps its performance even for large prediction errors. The reason is that instead of only detecting the segment type, our algorithm checks the cumulative differential cost and only turns on the generator when it sees enough benefits in the look-ahead window.
\vspace{-1mm}

\begin{figure} 
\begin{minipage}[b][1\totalheight][t]{0.48\columnwidth}%
\begin{center} 
\includegraphics[width=1\columnwidth]{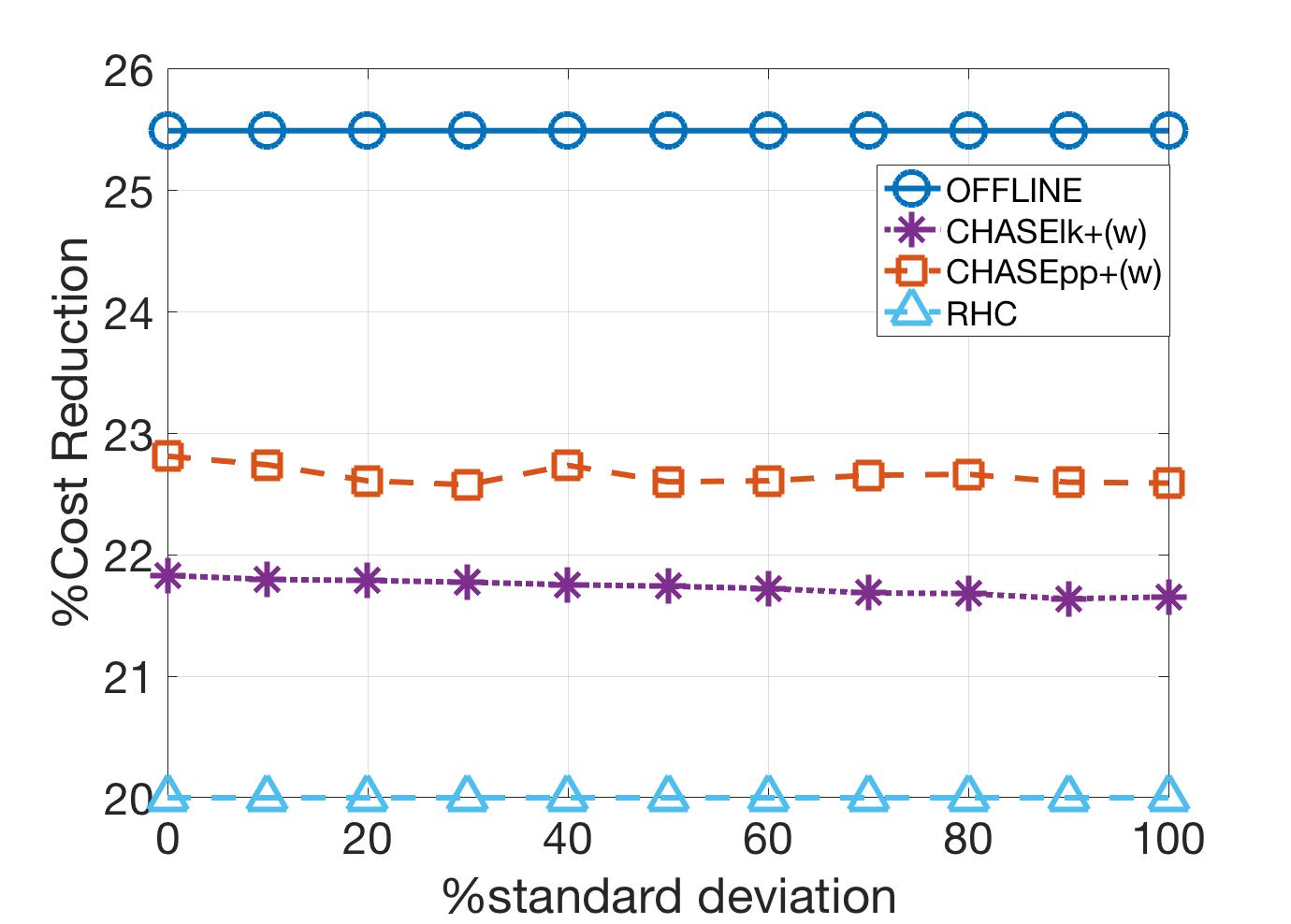}
\par\end{center}
\caption{\label{fig:crerror1}Cost reduction for different sizes of the normal prediction error ($w=1$).}
\end{minipage}\hfill{}%
\begin{minipage}[b][1\totalheight][t]{0.48\columnwidth}%
\begin{center} 
\includegraphics[width=1\columnwidth]{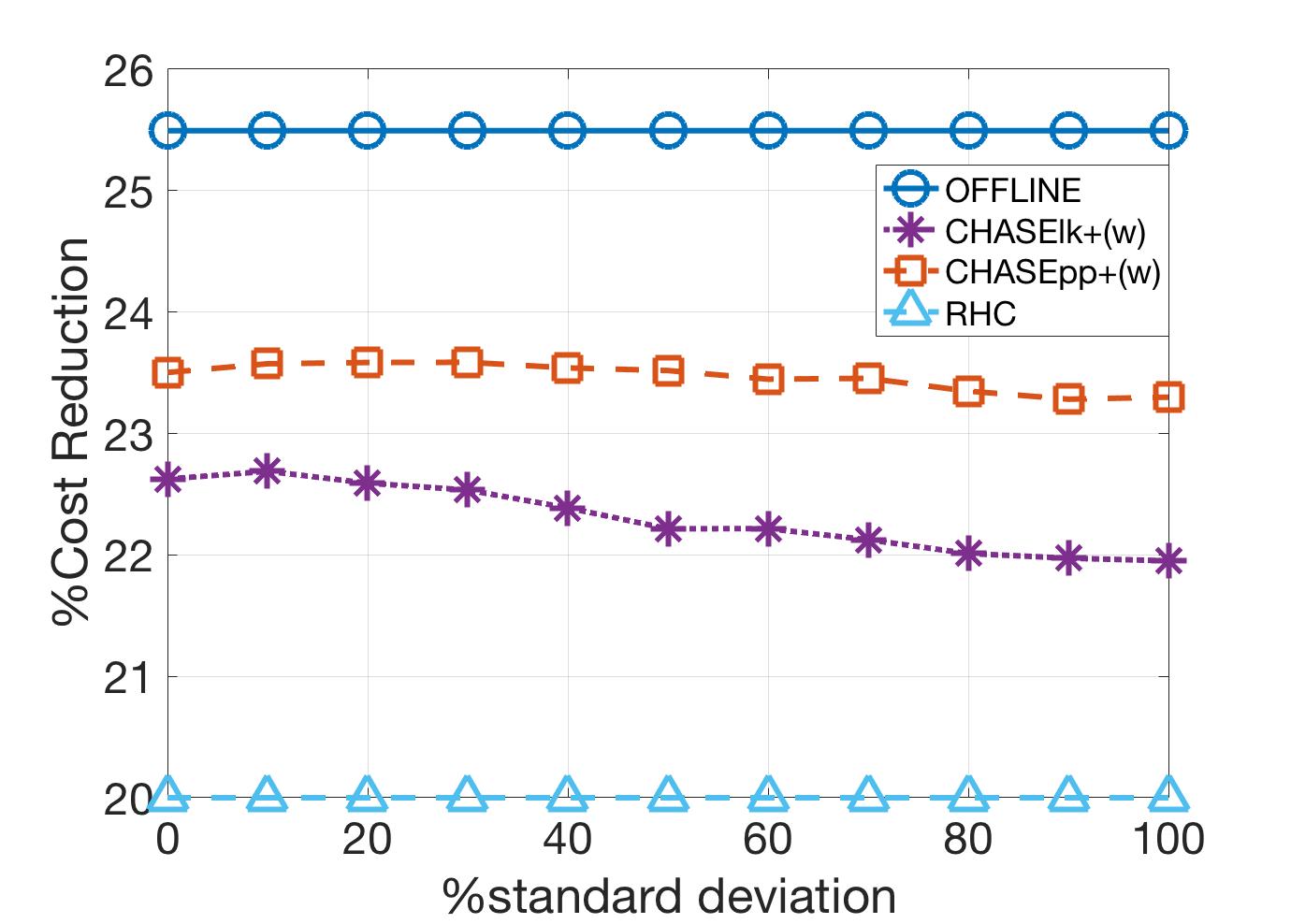}
\par\end{center} 
\caption{\label{fig:crerror3}Cost reduction for different sizes of the normal prediction error ($w=3$).} 
\end{minipage}
\begin{minipage}[b][1\totalheight][t]{0.48\columnwidth}%
\begin{center} 
\includegraphics[width=1\columnwidth]{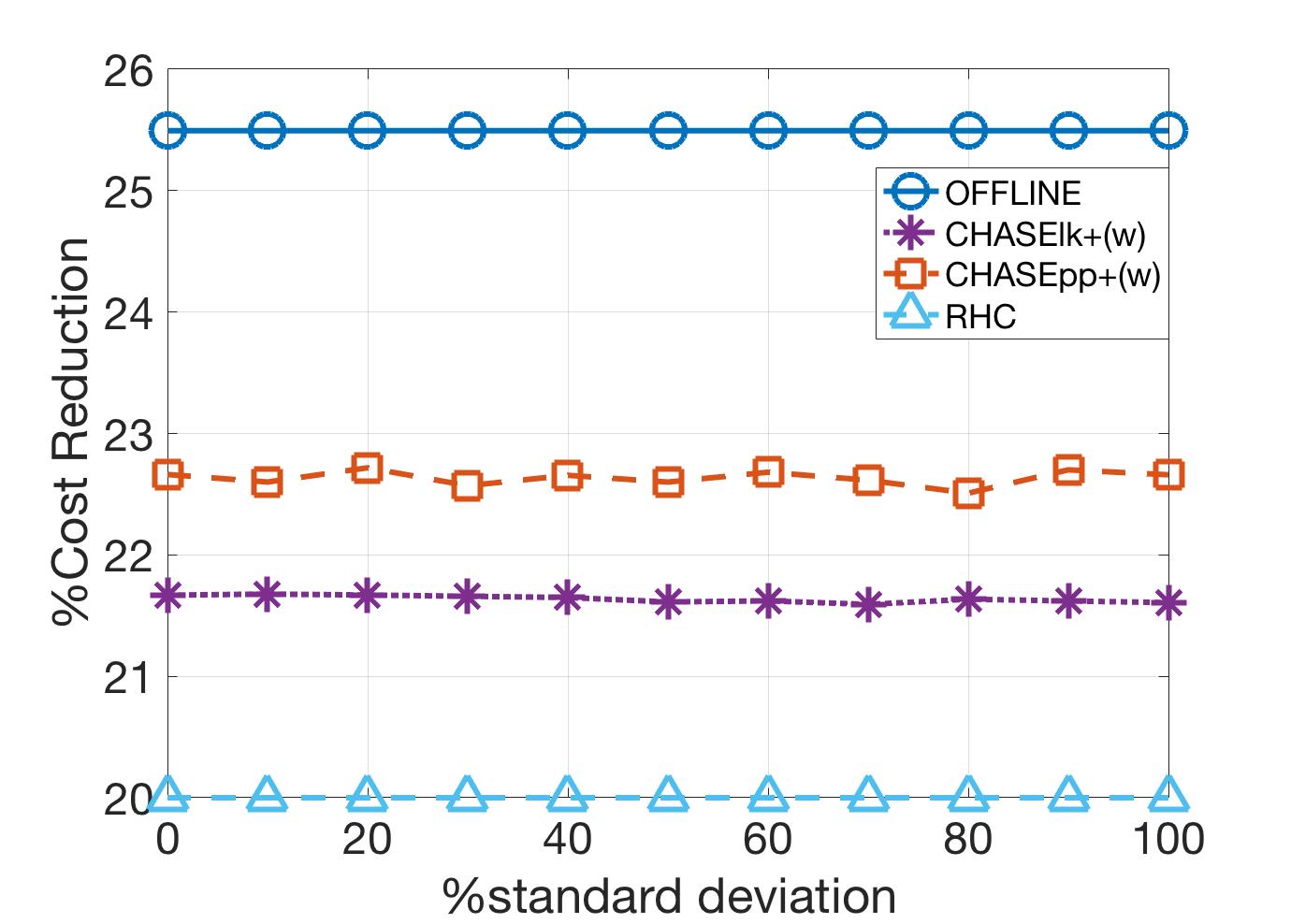}
\par\end{center}
\caption{\label{fig:hype1}Cost reduction for different sizes of the hyperbolic prediction error ($w=1$).}
\end{minipage}\hfill{}%
\begin{minipage}[b][1\totalheight][t]{0.48\columnwidth}%
\begin{center} 
\includegraphics[width=1\columnwidth]{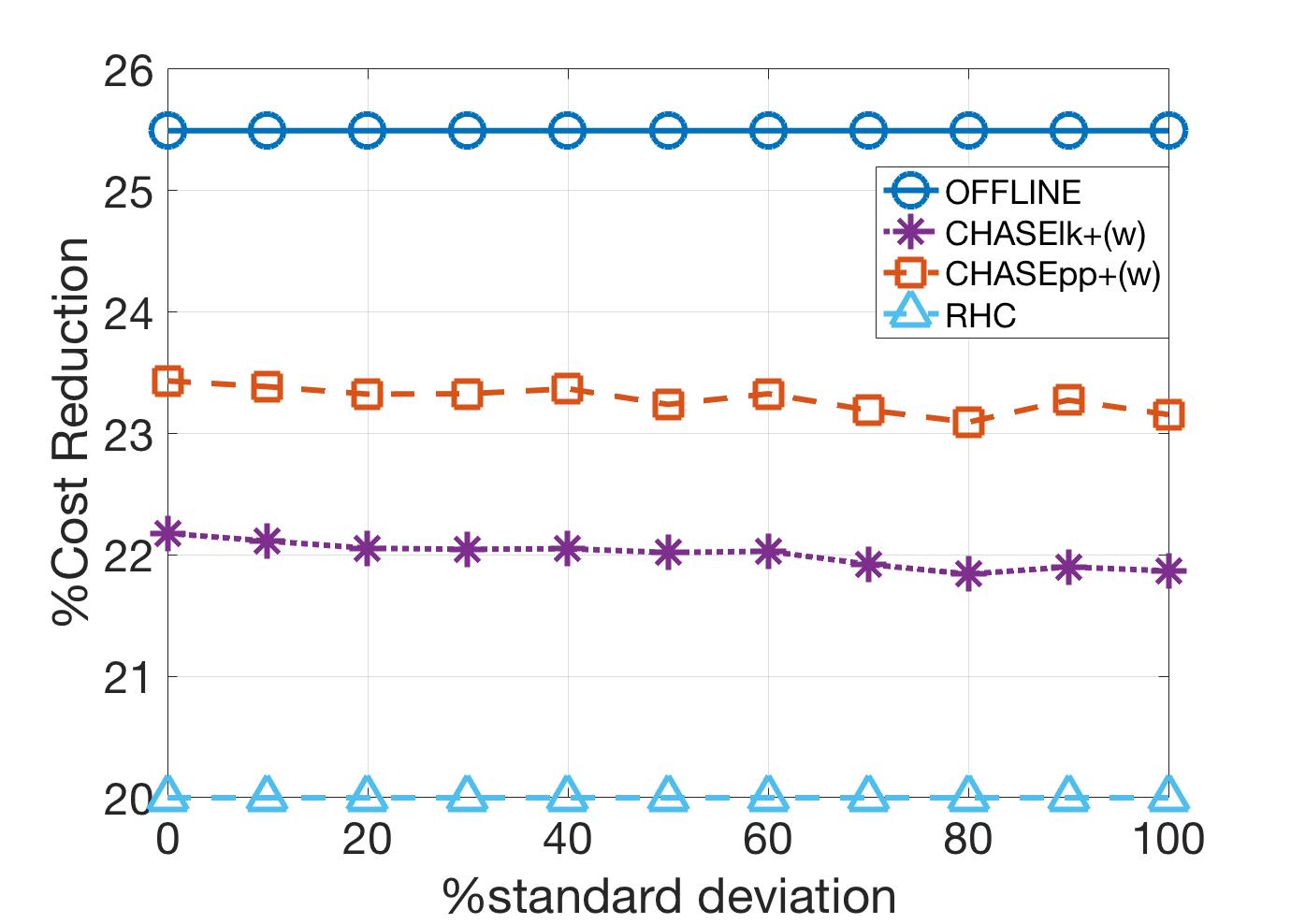}
\par\end{center} 
\caption{\label{fig:hype3}Cost reduction for different sizes of the hyperbolic prediction error ($w=3$).} 
\end{minipage}
\vspace{-4mm}
\end{figure}  
 
\vspace{1mm}
\section{Conclusion and Future Work}
\label{sec:conc}
We investigate how to leverage the prediction of the near future for online energy generation scheduling in microgrids. We tackle this problem from a new perspective and propose an effective online algorithm design that is fundamentally different from the existing algorithms. Our novel threshold-based online algorithm attains the best competitive ratio to date, which is upper-bounded by $3-2/(1+\mathcal{O}(\frac{1}{w}))$, where $w$ is the prediction window size. We also characterize a non-trivial lower bound of the competitive ratio and show that the competitive ratio of our algorithm is only $9\%$ away from the lower bound, when a few hours of prediction is available. Our theoretical and empirical evaluations demonstrate that our online algorithm outperforms the state-of-the-art ones. An interesting future direction is to exploit our new design space in developing competitive algorithms for general MTS problems with limited prediction. We also plan to incorporate energy storage system into the problem setting and algorithm design. 
\vspace{-1mm}
\section{Acknowledgement}
The work presented in this paper was supported in part by a Start-up Grant from School of Data Science (Project No. 9380118), City University
of Hong Kong, and a General Research Fund from Research Grants Council, Hong Kong (Project No. CityU 11206821).

\bibliographystyle{ACM-Reference-Format}
\bibliography{ref}

\appendix

\section{Proof of Theorem~\ref{lem:crfunction}}
\label{sec:appendix A}
\begin{proof} 
Without loss of generality let us consider that in a type 1 critical segment the algorithm turns on the generator at time $t= \tilde{T}_{i}^{c}+kw-\theta$, where $k \in [0,\infty)$ and $\theta \in [1,w]$.  $k = 0$ and $\theta = w$ gives us the case that we turn on the generator at $t=\tilde{T}_{i}^{c}-w$. Here we first consider the case with $k=0$, and then we extend the result to the general $k$. In the case with $k=0$ we turn on the generator at time $t=\tilde{T}_{i}^{c}-\theta$ which means we keep the generator off from $t=\tilde{T}_{i}^{c}-w$ to $t=\tilde{T}_{i}^{c}-\theta$, where $\theta \in [1,w]$.

We denote the outcome of ${\sf CHASEpp}(w)$ by $\big(y_{{\sf CHASE}(w)}(t)\big)_{t=1}^{T}$, and the outcome of the optimal offline algorithm by $\big(y_{OFA}(t)\big)_{t=1}^{T}$. Let us define ${\mathcal K}_{j}$ as the set of indices of type-$j$ critical segments i.e., 
\begin{equation}
   {\mathcal K}_{j} \triangleq  \{ \, i  \, |\, \, [T_i^c+1 , T_{i+1}^c] \textit{is a type-j critical segment in } [0, T] \}. \notag
\end{equation}
Denote the sub-cost for type-$j$ by
\begin{eqnarray}
{\rm Cost}^{{\rm ty\mbox{-}}j}(y)  \triangleq  \sum_{i\in{\mathcal K}_{j}}\sum_{t=T_{i}^{c}+1}^{T_{i+1}^{c}}\psi\big(\sigma(t),y(t)\big)+\beta\cdot[y(t)-y(t-1)]^{+}
\end{eqnarray}
Hence, ${\rm Cost}(y)=\sum_{j=0}^{3}{\rm Cost}^{{\rm ty\mbox{-}}j}(y)$.
We prove by comparing the sub-cost for each type-$j$.

(\textbf{Type-0}): Note that both $y_{{\rm OFA}}(t)=y_{{\rm CHASE(w)}}(t)=0$
for all $t\in[1,T]$. Hence,
\begin{equation}
{\rm Cost}^{{\rm ty\mbox{-}}0}(y_{{\rm OFA}})={\rm Cost}^{{\rm ty\mbox{-}}0}(y_{{\rm CHASE(w)}})
\end{equation}

(\textbf{Type-1}): Based on the definition of critical segment,
we recall that there is an auxiliary point $\tilde{T_{i}^{c}}$, such that either \big($\Delta(T_{i}^{c})=0$ and $\Delta(\tilde{T_{i}^{c}})=-\beta$\big)
or \big($\Delta(T_{i}^{c})=-\beta$ and $\Delta(\tilde{T_{i}^{c}})=0$\big). 
 We assumed we turn on the generator at $t= \Tilde{T_i^c}-\theta$. Now we focus on the segment $ [T_i^c+1, T_{i+1}^c]$. We observe
\begin{equation} 
y_{{\sf CHASE}(w)}(t)=\begin{cases}
0, & \mbox{for all\ } t\in[T_{i}^{c}+1,\tilde{T}_{i}^{c}-\theta),\\
1, & \mbox{for all\ } t\in[\tilde{T}_{i}^{c}-\theta,T_{i+1}^{c}].
\end{cases}
\end{equation}
We consider a particular type-1 critical segment: $[T_{i}^{c}+1,T_{i+1}^{c}]$. Note that by the definition
of type-1, at the beginning of this segment for both online and offline algorithms the generator status is off $y_{{\rm OFA}}(T_{i}^{c})=y_{{\sf CHASE}(w)}(T_{i}^{c})=0$. For the offline algorithm
$y_{{\rm OFA}}(t)$ switches from $0$ to $1$ at time $T_{i}^{c}+1$,
while for the online algorithm $y_{{\sf CHASE}(w)}$ switches at time $\tilde{T}_{i}^{c}-\theta-1$,
both incurring startup cost $\beta$. Hence, in the interval $[T_{i}^{c}+1,\tilde{T}_{i}^{c}-\theta-1]$, the online and offline algorithms have a different status, while in the interval $[\tilde{T}_{i}^{c}-\theta, T_{i+1}^{c}]$ they have the same status. The cost difference between $y_{{\sf CHASE}(w)}$ and $y_{{\rm OFA}}$ within $[T_{i}^{c}+1,T_{i+1}^{c}]$
is
\begin{eqnarray} \label{eq31}
&&\sum_{t=T_{i}^{c}+1}^{\tilde{T}_{i}^{c}-1}\Big(\psi\Big(\sigma(t),0\Big)-\psi\Big(\sigma(t),1\Big)\Big)+\beta-\beta = \sum_{t=T_{i}^{c}+1}^{\tilde{T}_{i}^{c}-\theta-1}\delta(t) \notag \\
&&= \Delta(\tilde{T}_{i}^{c}-\theta-1)-\Delta(T_{i}^{c})= -q_i^1+\beta,
\end{eqnarray}
where $q_i^1  \triangleq - \Delta(\tilde{T}_{i}^{c}-\theta-1)$. 

If we repeat this process for all type-1 critical segments, and we have $m_{1}\triangleq|{\mathcal K}_{1}|$ type-1 critical segments, we obtain
\begin{equation} \label{eq32}
{\rm Cost}^{{\rm ty\mbox{-}}1}(y_{{\sf CHASE}(w)})\le{\rm Cost}^{{\rm ty\mbox{-}}1}(y_{{\rm OFA}})+m_{1}\cdot\beta-\sum_{i\in {\mathcal K}_{1}}q_{i}^{1}.
\end{equation}
(\textbf{Type-2}) and (\textbf{Type-3}):
Now, we repeat the same process for type-2 and type-3 (type-end) critical segments. Let the number of type-$j$ critical segments be $m_{j}\triangleq|{\mathcal K}_{j}|$. We derive similarly for $j=2$ or $3$ as
\begin{equation}
{\rm Cost}^{{\rm ty\mbox{-}}j}(y_{{\sf CHASE}(w)})\le{\rm Cost}^{{\rm ty\mbox{-}}j}(y_{{\rm OFA}})+m_{j}\cdot \beta-\sum_{i\in {\mathcal K}_{j}}q_{i}^{j},
\end{equation}
where $q_{i}^{j}  \triangleq \beta+ \Delta(\tilde{T}_{i}^{c}-w-1)$. 

Note that $|q_{i}^{j}|\leq\beta$ for all $i,j$.
Furthermore, we note $m_{1}=m_{2}+m_{3}$, because it takes equal
numbers of critical segments for increasing $\Delta(\cdot)$ from
$-\beta$ to 0 and for decreasing from 0 to $-\beta$. We obtain

\begin{eqnarray}
&&{\displaystyle \frac{{\rm Cost}(y_{{\sf CHASE}(w)})}{{\rm Cost}(y_{{\rm OFA}})}}= {\displaystyle \frac{\sum_{j=0}^{3}{\rm Cost}^{{\rm ty\mbox{-}}j}(y_{{\sf CHASE}(w)})}{\sum_{j=0}^{3}{\rm Cost}^{{\rm ty\mbox{-}}j}(y_{{\rm OFA}})}} \leq  \notag\\
&& {\displaystyle  1+ \frac{2m_{1}\beta+\sum_{i\in {\mathcal K}_{1}}q_{i}^{1}-\sum_{i\in {\mathcal K}_{2}}q_{i}^{2}-\sum_{i\in {\mathcal K}_{3}}q_{i}^{3} }{\sum_{j=0}^{3}{\rm Cost}^{{\rm ty\mbox{-}}j}(y_{{\rm OFA}})}}
\end{eqnarray}
It should be noted that in the calculation a type-3 critical segment is exactly the same as a type-2 critical segment and hence in the rest of the calculation we just consider type-2 critical segments. As a result we can write $m_{1}=m_{2}$ for the ease of calculation. We have

\begin{eqnarray}
&  &{\displaystyle \frac{{\rm Cost}(y_{{\sf CHASE}(w)})}{{\rm Cost}(y_{{\rm OFA}})}} \leq  {\displaystyle  1+\frac{2m_{1}\beta-\sum_{i\in {\mathcal K}_{1}}q_{i}^{1}-\sum_{i\in {\mathcal K}_{2}}q_{i}^{2}}{\sum_{j=0}^{2}{\rm Cost}^{{\rm ty\mbox{-}}j}(y_{{\rm OFA}})}}\notag \\ 
&  & \leq 1+\begin{cases}
0 & \mbox{if\ }m_{1}=0,\\
{\displaystyle \frac{2m_{1}\beta-\sum_{i\in {\mathcal K}_{1}}q_{i}^{1} -\sum_{i\in {\mathcal K}_{2}}q_{i}^{2}}{\sum_{j=0}^{2}{\rm Cost}^{{\rm ty\mbox{-}}j}(y_{{\rm OFA}})}} & \mbox{otherwise}
\end{cases} 
\end{eqnarray}
By Lemma \ref{lem:lower bound type-1}, and Lemma \ref{lem:lower bound type-2} and simplifications, we obtain: 
\begin{eqnarray}
&& \frac{{\rm Cost}(y_{{\sf CHASE}(w)})}{{\rm Cost}(y_{{\rm OFA}})} \leq  1+ \big(1- \frac{Lc_{\mathrm{o}}+c_{\mathrm{m}}}{L(p_{\mathrm{max}}+\eta\cdot c_{\mathrm{g}})} \big) \cdot \max\limits_{{q \in \{0,wc_{\mathrm{m}}\}}}  \notag \\
&& \big\{ \frac{ (2\beta-q)}{  \beta+\big(2wc_{\mathrm{m}}-q+\frac{c_{\mathrm{o}}}{p_{\mathrm{max}}+\eta\cdot c_{\mathrm{g}}}\lambda\big) \big(1-\frac{c_{\mathrm{m}}}{L(p_{\mathrm{max}}+\eta\cdot c_{\mathrm{g}}-c_{\mathrm{o}})} \big) } \big\} 
\end{eqnarray} 
We denote this value as $R_{\mathrm{on}}(\lambda)$. This is the performance ratio for the case with $k=0$. Now if we increase $k$, by using the same process we can see that both online and offline cost increase. Using the result from Lemma \ref{lem:Roff}, if we change to general $k$, the increment ratio is as follows:
\begin{eqnarray}
\frac{\text{online cost increment}}{\text{offline cost increment}} \leq \frac{k\big( wc_{\mathrm{m}}+\lambda \big)}{k \big( wc_{\mathrm{m}}+\frac{c_{\mathrm{o}}}{p_{\mathrm{max}}+\eta\cdot c_{\mathrm{g}}}\lambda \big)}   
\end{eqnarray}  
We denote this ratio as $R_{\mathrm{off}}(\lambda)$. If this ratio is larger than the previous one $R_{\mathrm{off}}(\lambda) > R_{\mathrm{on}}(\lambda)$, then by increasing $k$, the value of the performance ratio keeps increasing and when $k$ goes to $\infty$, this competitive ratio goes to $R_{\mathrm{off}}(\lambda)$. On the other hand, if $R_{\mathrm{off}}(\lambda)\leq R_{\mathrm{on}}(\lambda)$, by increasing the value of $k$, value of the competitive ratio will not increase and is still upper bounded by $R_{\mathrm{on}}(\lambda)$. This shows that the competitive ratio is upper bonded by the maximum of $R_{\mathrm{on}}(\lambda)$, and $R_{\mathrm{off}}(\lambda)$. In Lemma~\ref{lem:incresing} we show that $R_{\mathrm{off}}(\lambda)$ is always an increasing function while $R_{\mathrm{on}}(\lambda)$ is always a decreasing function.  This completes the proof.
\end{proof}

\subsection{Proof of Lemma~\ref{lem:lower bound type-1}}

\begin{lem} For the (\textbf{type-1}), we have
\label{lem:lower bound type-1}
\begin{eqnarray}
 &  & {\rm Cost}^{{\rm ty\mbox{-}}1}(y_{{\rm OFA}}) \geq m_{1}\beta+\sum_{i\in {\mathcal K}_{1}}\Big(\frac{(q_{i}^{1}+\beta)(Lc_{\mathrm{o}}+c_{\mathrm{m}})}{L\big(p_{\mathrm{max}}+\eta\cdot c_{\mathrm{g}}-c_{\mathrm{o}}\big)-c_{\mathrm{m}}} + \notag\\
 &  & \frac{p_{\mathrm{max}}+\eta\cdot c_{\mathrm{g}}}{p_{\mathrm{max}}+\eta\cdot c_{\mathrm{g}}-c_{\mathrm{o}}} \big( wc_{\mathrm{m}} + \frac{c_{\mathrm{o}}}{p_{\mathrm{max}}+\eta\cdot c_{\mathrm{g}}} \lambda \big) \Big). 
\end{eqnarray}
\end{lem}

\begin{proof}
Consider a particular type-1 segment $[T_{i}^{c}+1,T_{i+1}^{c}]$. We denote its offline cost as $ \mathrm{Cost^{\rm t1}} $. We have: 
\begin{eqnarray}
&& \mathrm{Cost^{\rm t1}}  =  \beta+\sum_{t=T_{i}^{c}+1}^{T_{i+1}^{c}}\psi\big(\sigma(t),1\big) = \beta+\notag \\
 && (T_{i+1}^{c}-T_{i}^{c})c_{\mathrm{m}}\label{eq:type-1 cost eq1-1} +\sum_{t=T_{i}^{c}+1}^{T_{i+1}^{c}}\big(\psi\big(\sigma(t),1\big)-c_{\mathrm{m}}\big).
\end{eqnarray} 
By \cite[Lemma.~4]{Minghua2013SIG} and simplification we obtain
\begin{eqnarray} 
&  &  \mathrm{Cost^{\rm t1}} \geq  \beta+(T_{i+1}^{c}-T_{i}^{c})c_{\mathrm{m}}+ \frac{c_{\mathrm{o}}}{p_{\mathrm{max}}+\eta\cdot c_{\mathrm{g}}-c_{\mathrm{o}}} \cdot \notag \\ 
&& \Big(\sum_{t=T_{i}^{c}+1}^{T_{i+1}^{c}} \delta(t)+(T_{i+1}^{c}-T_{i}^{c})c_{\mathrm{m}}\Big) =  \beta+ \frac{p_{\mathrm{max}}+\eta\cdot c_{\mathrm{g}}}{p_{\mathrm{max}}+\eta\cdot c_{\mathrm{g}}-c_{\mathrm{o}}} \cdot \notag \\
 & &  (T_{i+1}^{c}-T_{i}^{c})c_{\mathrm{m}}   + \frac{c_{\mathrm{o}}}{p_{\mathrm{max}}+\eta\cdot c_{\mathrm{g}}-c_{\mathrm{o}}} \sum_{t=T_{i}^{c}+1}^{T_{i+1}^{c}} \delta(t)  
\end{eqnarray}

Now we need to find the lower bound of both $(T_{i+1}^{c}-T_{i}^{c})$ and $\sum_{t=T_{i}^{c}+1}^{T_{i+1}^{c}} \delta(t)$ in the following two steps.

{\bf Step 1:} We write the lower bound of $\sum_{t=T_{i}^{c}+1}^{T_{i+1}^{c}} \delta(t)$ as follows:
\begin{eqnarray} 
&& \sum_{t=T_{i}^{c}+1}^{T_{i+1}^{c}} \delta(t) = \sum_{t=T_{i}^{c}+1}^{\Tilde{T_i^c}-\theta -1} \delta(t)+ \sum_{\Tilde{T_i^c}-\theta }^{T_{i+1}^{c}} \delta(t) = \Delta(\Tilde{T_i^c}-\theta -1)- \Delta(T_{i}^{c})\notag \\
&&+  \sum_{\Tilde{T_i^c}-\theta }^{T_{i+1}^{c}} \delta(t) =  \beta - q_i^1  + \sum_{\Tilde{T_i^c}-\theta }^{T_{i+1}^{c}} \delta(t)\geq \beta - q_i^1 +\lambda  \notag 
\label{eq:w type-1 delta lowebound}
\end{eqnarray} 
{\bf Step 2:} To find the lower bound of the length of the interval $[T_{i}^{c}+1, T_{i+1}^{c}]$, we have two cases with $\theta = w$ or $\theta < w$ we calculate the lower bound as follows: 

{\bf Case 1:} If $\theta = w$,  we can see that  $[T_{i}^{c}+1, T_{i+1}^{c}]$ has two part as  $[T_{i}^{c}+1, \Tilde{T_i^c}-w-1]$ and $[\Tilde{T_i^c}-w, T_{i+1}^{c}]$. We note that $\big(\Tilde{T_i^c}-w-1-T_{i}^{c}\big)$
is lower bounded by the steepest descend when $p(t)=p_{\mathrm{max}}$, $a(t)=L$
and $h(t)=\eta L$,
 \begin{eqnarray}
  \Tilde{T_i^c}-w-1-T_{i}^{c} \geq \frac{\beta - q_i^1}{L\big(p_{\mathrm{max}}+\eta\cdot c_{\mathrm{g}}-c_{\mathrm{o}}\big)-c_{\mathrm{m}}},
  \label{eq:w type-1 cost eq7}
 \end{eqnarray}
 and for the second part we have $  T_{i+1}^{c}- \Tilde{T_i^c}+w +1 \geq  w,$ which means
 \begin{eqnarray}
  T_{i+1}^{c}- T_{i}^{c} \geq \frac{\beta - q_i^1}{L\big(p_{\mathrm{max}}+\eta\cdot c_{\mathrm{g}}-c_{\mathrm{o}}\big)-c_{\mathrm{m}}}+ w
 \end{eqnarray}
 
 {\bf Case 2:} On the other hand, when $\theta < w$, we know the length of the interval $[\Tilde{T_i^c}-w,\Tilde{T_i^c} ]$ is $w+1$ time slot and its cost difference is less than $\lambda$, hence to calculate the total $[T_{i}^{c}+1, T_{i+1}^{c}]$ length we have

\begin{eqnarray}
&& \sum_{t=T_{i}^{c}+1}^{T_{i+1}^{c}} \delta(t) \geq \beta - q_i^1 +\lambda \implies \sum_{t=T_{i}^{c}+1}^{T_{i+1}^{c}} \delta(t)  - \sum_{\Tilde{T_i^c}-w}^{\Tilde{T_i^c}} \delta(t) +  \\
&& \sum_{\Tilde{T_i^c}-w}^{\Tilde{T_i^c}} \delta(t)\geq \beta - q_i^1 +\lambda  \implies \sum_{t=T_{i}^{c}+1}^{T_{i+1}^{c}} \delta(t)  - \sum_{\Tilde{T_i^c}-w}^{\Tilde{T_i^c}} \delta(t) \geq \beta - q_i^1, \notag
\end{eqnarray}
where the last inequality comes from the fact that $\sum_{\Tilde{T_i^c}-w}^{\Tilde{T_i^c}} \delta(t) \leq \lambda$. We note that $\big(T_{i+1}^{c}-T_{i}^{c} - (w+1)\big)$
is lower bounded by the steepest descend when $p(t)=p_{\mathrm{max}}$, $a(t)=L$
and $h(t)=\eta L$,
\begin{eqnarray}
&& T_{i+1}^{c}-T_{i}^{c} - (w+1) \geq \frac{\beta - q_i^1}{L\big(p_{\mathrm{max}}+\eta\cdot c_{\mathrm{g}}-c_{\mathrm{o}}\big)-c_{\mathrm{m}}}  \notag \\
&& \implies T_{i+1}^{c}-T_{i}^{c} \geq \frac{\beta - q_i^1}{L\big(p_{\mathrm{max}}+\eta\cdot c_{\mathrm{g}}-c_{\mathrm{o}}\big)-c_{\mathrm{m}}}+w
\end{eqnarray}
So one can see that in both of these cases we always have
\begin{eqnarray}
T_{i+1}^{c}-T_{i}^{c} \geq \frac{\beta - q_i^1}{L\big(p_{\mathrm{max}}+\eta\cdot c_{\mathrm{g}}-c_{\mathrm{o}}\big)-c_{\mathrm{m}}}+w
\label{eq:w type-1 length lowebound}
\end{eqnarray}
length
By Eqns.~\eqref{eq:w type-1 delta lowebound}-\eqref{eq:w type-1 length lowebound}, we obtain

\begin{eqnarray}
&& \mathrm{Cost^{\rm t1}} \geq  \beta+   \notag \\
 &&\frac{p_{\mathrm{max}}+\eta\cdot c_{\mathrm{g}}}{p_{\mathrm{max}}+\eta\cdot c_{\mathrm{g}}-c_{\mathrm{o}}} \big( \frac{\beta - q_i^1}{L(p_{\mathrm{max}}+\eta c_{\mathrm{g}}-c_{\mathrm{o}}-\frac{c_{\mathrm{m}}}{L})}+w \big)c_{\mathrm{m}} +  \notag \\
 && \frac{c_{\mathrm{o}}}{p_{\mathrm{max}}+\eta\cdot c_{\mathrm{g}}-c_{\mathrm{o}}} \big(\beta - q_i^1+\lambda \big)=  \beta+ \frac{(\beta - q_i^1)(Lc_{\mathrm{o}}+c_{\mathrm{m}})}{L(p_{\mathrm{max}}+\eta c_{\mathrm{g}}-c_{\mathrm{o}}-\frac{c_{\mathrm{m}}}{L})}  \notag \\
 && + \frac{p_{\mathrm{max}}+\eta\cdot c_{\mathrm{g}}}{p_{\mathrm{max}}+\eta\cdot c_{\mathrm{g}}-c_{\mathrm{o}}} \big(wc_{\mathrm{m}} + \frac{c_{\mathrm{o}}}{p_{\mathrm{max}}+\eta\cdot c_{\mathrm{g}}} \lambda \big) 
 \label{eq:w type-1 cost eq6}  
\end{eqnarray} 

Since there are $m_{1}$ type-1 critical segments, according to Eqna.~\eqref{eq:w type-1 cost eq6}, we obtain
\begin{eqnarray}
 &  & {\rm Cost}^{{\rm ty\mbox{-}}1}(y_{{\rm OFA}}) \geq m_{1}\beta+\sum_{i\in {\mathcal K}_{1}}\Big(\frac{(\beta- q_i^1)(Lc_{\mathrm{o}}+c_{\mathrm{m}})}{L\big(p_{\mathrm{max}}+\eta\cdot c_{\mathrm{g}}-c_{\mathrm{o}}\big)-c_{\mathrm{m}}} \notag\\
 &&+\frac{p_{\mathrm{max}}+\eta\cdot c_{\mathrm{g}}}{p_{\mathrm{max}}+\eta\cdot c_{\mathrm{g}}-c_{\mathrm{o}}} \big( wc_{\mathrm{m}} + \frac{c_{\mathrm{o}}}{p_{\mathrm{max}}+\eta\cdot c_{\mathrm{g}}} \lambda \big) \Big). 
\end{eqnarray}
\end{proof}

\subsection{Proof of Lemma~\ref{lem:lower bound type-2}}

\begin{lem}  For (\textbf{type-2}), we have
\label{lem:lower bound type-2}
\begin{eqnarray}
{\rm Cost}^{{\rm ty\mbox{-}}2}(y_{{\rm OFA}}) \geq  \sum_{i\in {\mathcal K}_{2}}\Big(\frac{p_{\mathrm{max}}+\eta\cdot c_{\mathrm{g}}}{p_{\mathrm{max}}+\eta\cdot c_{\mathrm{g}}-c_{\mathrm{o}}} (w \cdot c_{\mathrm{m}}-q_{i}^{2}) \Big)
\end{eqnarray}
\end{lem}

\begin{proof}
Consider a particular type-2 segment $[T_{i}^{c}+1,T_{i+1}^{c}]$. We denote its offline cost as $ \mathrm{Cost^{\rm t2}} $. We have:
\begin{eqnarray}
\mathrm{Cost^{\rm t2}} = \sum_{t=T_{i}^{c}+1}^{T_{i+1}^{c}}\psi\big(\sigma(t),0\big) 
\label{eq:w type-2 cost eq1}
\end{eqnarray} 
By \cite[Lemma.~4]{Minghua2013SIG} and simplification we obtain
\begin{eqnarray}
 &  & \mathrm{Cost^{\rm t2}} \geq  \frac{p_{\mathrm{max}}+\eta\cdot c_{\mathrm{g}}}{p_{\mathrm{max}}+\eta\cdot c_{\mathrm{g}}-c_{\mathrm{o}}} \bigg( \sum_{t=T_{i}^{c}+1}^{T_{i+1}^{c}}\delta(t)+ (T_{i+1}^{c}-T_{i}^{c})c_{\mathrm{m}} \bigg)\notag \\
&& = \frac{p_{\mathrm{max}}+\eta\cdot c_{\mathrm{g}}}{p_{\mathrm{max}}+\eta\cdot c_{\mathrm{g}}-c_{\mathrm{o}}} \bigg( -\beta + (T_{i+1}^{c}-T_{i}^{c})c_{\mathrm{m}} \bigg) 
\label{eq:w type-2 cost eq2}
\end{eqnarray}
Furthermore, we note that $\big(T_{i+1}^{c}-T_{i}^{c}\big)$
is lower bounded by the steepest descend when $\min \{a(t),h(t)\}=0$,
\begin{equation}
T_{i+1}^{c}-T_{i}^{c}\geq w + \frac{\beta - q_{i}^{2}}{c_{\mathrm{m}}}
\label{eq:w type-2 cost eq3}
\end{equation}

By Eqns.~(\ref{eq:w type-2 cost eq2})-(\ref{eq:w type-2 cost eq3}), we obtain
\begin{eqnarray}
\mathrm{Cost^{\rm t2}} & \geq & \frac{p_{\mathrm{max}}+\eta\cdot c_{\mathrm{g}}}{p_{\mathrm{max}}+\eta\cdot c_{\mathrm{g}}-c_{\mathrm{o}}} ( w \cdot c_{\mathrm{m}} -q_{i}^{2})
\label{eq:w type-2 cost roof-2}
\end{eqnarray}
Since there are $m_{2}$ type-2 critical segments, according to Eqna.~(\ref{eq:w type-2 cost  roof-2}), we obtain
\begin{eqnarray}  
 {\rm Cost}^{{\rm ty\mbox{-}}2}(y_{{\rm OFA}}) \geq  \sum_{i\in {\mathcal K}_{2}}\Big(\frac{p_{\mathrm{max}}+\eta\cdot c_{\mathrm{g}}}{p_{\mathrm{max}}+\eta\cdot c_{\mathrm{g}}-c_{\mathrm{o}}} ( w \cdot c_{\mathrm{m}} -q_{i}^{2})\Big). 
\end{eqnarray}
\end{proof}

\subsection{Proof of Lemma~\ref{lem:Roff}}
\label{sec:appendix B}

\begin{lem}
Consider a window $[t,t+w]$, in the type-1 critical segment. If we have $\Delta_{t}^{t+w} \leq \lambda $, then the cost of the online algorithm over the cost of the optimal offline algorithm in this window is upper bounded by the following:
\label{lem:Roff}
	\begin{eqnarray} 
	 \frac{{\rm Cost}(y_{{\sf CHASEpp}(w)})}{{\rm Cost}(y_{{\rm OFA}})}\leq  \frac{wc_{\mathrm{m}}+\lambda}{wc_{\mathrm{m}}+\frac{c_{\mathrm{o}}}{p_{\mathrm{max}}+\eta\cdot c_{\mathrm{g}}}\lambda}.
	\end{eqnarray}  
\end{lem} 

\begin{proof}
By \cite[Lemma.~4]{Minghua2013SIG} and simplification we know that in a type-1 critical segment for a window with $\Delta_{t}^{t+w}=\lambda$, for the offline cost we have
\begin{eqnarray}
 {\rm Cost}(y_{{\rm OFA}}) \geq \frac{p_{\mathrm{max}}+\eta\cdot c_{\mathrm{g}}}{p_{\mathrm{max}}+\eta\cdot c_{\mathrm{g}}-c_{\mathrm{o}}}\big( wc_{\mathrm{m}}+\frac{c_{\mathrm{o}}}{p_{\mathrm{max}}+\eta\cdot c_{\mathrm{g}}}\lambda \big) \notag\\
\end{eqnarray}
On the other hand, if in the type-1 critical segment the online keep the generator off in this window, the cost difference between the online and the offline is equal to $\lambda$, which means 
\begin{eqnarray}
{\rm Cost}(y_{{\sf CHASEpp}(w)}) - {\rm Cost}(y_{{\rm OFA}}) =\lambda 
\end{eqnarray}
Hence we have
\begin{eqnarray}
\frac{{\rm Cost}(y_{{\sf CHASEpp}(w)})}{{\rm Cost}(y_{{\rm OFA}})} = 1+\frac{\lambda}{{\rm Cost}(y_{{\rm OFA}})} \leq \frac{wc_{\mathrm{m}}+\lambda}{wc_{\mathrm{m}}+\frac{c_{\mathrm{o}}}{p_{\mathrm{max}}+\eta\cdot c_{\mathrm{g}}}\lambda}
\end{eqnarray}
This completes the proof.  
\end{proof}

\subsection{Proof of Lemma~\ref{lem:incresing}}

\begin{lem}  $R_{\mathrm{on}}(a)$ is always a decreasing function of $a$ and $R_{\mathrm{off}}(a)$ is always an increasing function of $a$.
\label{lem:incresing}
\end{lem} 

\begin{proof}
$R_{\mathrm{on}}(\lambda)$ : To prove that $R_{\mathrm{on}}(\lambda)$ is always a decreasing function first we take the derivative as a function of $\lambda$. To compute the derivative we only consider the maximization part of $R_{\mathrm{on}}(\lambda)$. The denominator of the function is always positive and the numerator is given by
\begin{eqnarray}
 - \big( \frac{(2\beta - q)c_{\mathrm{o}}}{p_{\mathrm{max}}+\eta\cdot c_{\mathrm{g}}}\lambda \big) \big(1-\frac{c_{\mathrm{m}}}{L(p_{\mathrm{max}}+\eta\cdot c_{\mathrm{g}}-c_{\mathrm{o}})} \big). 
\end{eqnarray}
Note that we have $q \in \{0,wc_{\mathrm{m}}\}$. If $wc_{\mathrm{m}} < 2 \beta$ the derivative is always negative. If $wc_{\mathrm{m}} \geq 2 \beta$ we show that we have $q=0$ in the maximization part of the function which again shows that derivative is negative. To show that for $wc_{\mathrm{m}} \geq 2 \beta$ we have $q=0$, we first take the derivative as a function of $q$ and we can see that the denominator of the function is always positive and the numerator is given by
\begin{eqnarray}
-\beta - \big( 2(wc_{\mathrm{m}}-\beta) + \frac{c_{\mathrm{o}}}{p_{\mathrm{max}}+\eta\cdot c_{\mathrm{g}}}\lambda \big) \big(1-\frac{c_{\mathrm{m}}}{L(p_{\mathrm{max}}+\eta\cdot c_{\mathrm{g}}-c_{\mathrm{o}})} \big). 
\end{eqnarray}
We can see that for $wc_{\mathrm{m}} \geq \beta$ value of the derivative is always negative which means that for $w c_{\mathrm{m}} \geq 2\beta$ the maximum value of $R_{\mathrm{on}}(\lambda)$ happens at $q=0$. This prove that for both cases the derivative is negative and hence $R_{\mathrm{on}}(\lambda)$ is always a decreasing function.

$R_{\mathrm{off}}(\lambda)$: Now we show that $R_{\mathrm{off}}(\lambda)$ is an increasing function. We take the derivative and we can see that the denominator of the function is always positive and the numerator is given by
\begin{eqnarray}
 \frac{p_{\mathrm{max}}+\eta\cdot c_{\mathrm{g}}-c_{\mathrm{o}}}{p_{\mathrm{max}}+\eta\cdot c_{\mathrm{g}}}wc_{\mathrm{m}},
\end{eqnarray}
which is always positive, and hence $R_{\mathrm{off}}(\lambda)$ is always an increasing function. This completes the proof.
\end{proof}

\section{Proof of Theorem~\ref{ref: competitive ratio}}
\label{sec:appendix C}
\begin{proof}
From Theorem~\ref{lem:crfunction} we have
\begin{eqnarray} 
{\sf CR}({\mathcal A}(\lambda^*)) = \max \{R_{\mathrm{on}}(\lambda^*), R_{\mathrm{off}}(\lambda^*)\}  
\end{eqnarray} 
and from the definition of the optimal threshold $\lambda^*$ in~\eqref{C_a_optimal} we have
\begin{eqnarray} 
\max \{R_{\mathrm{on}}(\lambda^*), R_{\mathrm{off}}(\lambda^*)\} = R_{\mathrm{on}}(\lambda^*)
\end{eqnarray} 
Therefore for the competitive ratio we have:
\begin{eqnarray} 
{\sf CR}({\mathcal A}(\lambda^*)) =  R_{\mathrm{on}}(\lambda^*).
\end{eqnarray} 
By using the definition of $\alpha$ in~\eqref{def:alpha} and simplification we obtain the result which completes the proof. 
\end{proof}

\section{Proof of Theorem~\ref{thm:nOFA-optimal}}
\label{sec:appendix M}
\begin{proof}

This theorem includes two parts as follows: 
\begin{itemize}
    \item \textbf{Offline algorithm:} As shown in \cite[Theorem.~5, and 6]{Minghua2013SIG}, when we have homogeneous generators, the offline algorithm that uses the layering approach produces an optimal offline solution for \textbf{MCMP}. For the case with multiple heterogeneous generators, since assigning the bottom layers to the generators with larger capacities minimizes the start-up cost, it can easily be shown that the layering approach also leads to an optimal offline solution. On the other hand, the operational cost does not depend on the capacity $L_n$, and it is the same for all the generators. Hence, it can be shown that the layering approach produces an offline optimal. 
    \item \textbf{Online algorithm:} In this case, each generator is solving its own sub-problem with a given sub-demand. Hence, the competitive ratio of the algorithm is upper bounded by the largest competitive ratio among all generators. 

\begin{equation}
 {\sf CR} \leq \max\limits_{{n \in [1,N] }}  3-2g(\alpha_n, w), 
\end{equation}
where 
\begin{equation}
\alpha_n= \frac{c_{\mathrm{o}}+c_{\mathrm{m}}/L_n}{p_{\mathrm{max}}+\eta c_{\mathrm{g}}}.
\end{equation}
Since the competitive ratio is an increasing function of $L$ and $L_1 \geq L_2 ... \geq L_{\mathrm{N}}$,  we have:
\begin{equation}
 {\sf CR} \leq  3-2g(\alpha_1, w).
\end{equation}
\end{itemize}

This completes the proof. 

\end{proof}

\section{Proof of Theorem~\ref{CRLOB}}
\label{sec:appendix L}
\begin{proof}
Finding the lower bound of the competitive ratio is equal to constructing an special input $\sigma(t) \triangleq (a(t), h(t), p(t))$ such that for any deterministic online algorithm $\mathcal A$, we have:
\begin{equation}
  \frac{{\rm Cost}({\sf y}_{\mathcal A};\sigma)}{{\rm Cost}({\sf y}_{\rm OFA},\sigma)}     \geq {\sf cr}(w).
\end{equation} 
In \cite{Minghua2013SIG}, it is shown that when we do not have any prediction $w=0$, the following input gives us the lower bound. 
\begin{eqnarray}
\delta(t)=
\begin{cases}
			\delta_{max}, & \mathrm{if } \quad y(t-1)=0 ,\\
			\delta_{min}, &  \mathrm{if } \quad y(t-1)=1.
		\end{cases}
\end{eqnarray}
As one can see in this input, as long as the algorithm keeps the generator off, the adversary keeps giving full demand ($\delta_{max}$) as the input, and as soon as the algorithm turns off the generator, the adversary starts giving zero demand ($\delta_{min}$) as the input.
This simple input is designed in a way that it always tries to hurt the algorithm most. In creating the lower bound for the case with perfect prediction, we follow the same logic, but we need to design a different input. 

If we keep giving full demand until the algorithm turns on the generator, at some point, we have already given a lot of full demand to the algorithm, and by turning on the generator, the algorithm can enjoy a window of full demand. In this way, we can not really hurt the algorithm. Therefore we need to carefully choose the demand in the future window in a way that it gives the algorithm some incentive to turn on the generator, but at the same time, it does not give it a lot of demand in the coming window to enjoy. By carefully adjusting this demand, we can find the lower bound.

At any time $t$, we need to construct the input of the time $t$, and the algorithm has already decided the generator status for the time $[1,t-w-1]$. We need to know that for how many consecutive time slots the algorithm kept the generator off. To this end, we define a counter called $c(t)$. This counter will reset anytime we turn on the generator and keeps increasing when the algorithm keeps the generator off. We define it as follows:

\begin{eqnarray}
c(t)= 
\begin{cases}
			0, & \mathrm{if } \quad y(t-w-1)=1,\\
			c(t-1)+1, &  \mathrm{if } \quad y(t-w-1)=0,
		\end{cases}
\end{eqnarray}
where the initial value of $y$ we have $c(0)=0$, and $y(t)_{t = -w}^{0}=0$.

We construct the worst-case input as follows:
\begin{eqnarray}
\delta(t)=
\begin{cases}
			\delta_1, & \mathrm{if } \quad c(t) \leq \frac{\beta -w \delta_2}{\delta_1} ,\\
		\delta_2, &  \mathrm{if } \quad c(t) > \frac{\beta -w \delta_2}{\delta_1}.
		\end{cases}
\end{eqnarray}
Now we need to calculate the proper values for $\delta_1$ and $\delta_2$. We use Lemma~\ref{prlb} toward this end.

\subsection{Proof of Lemma~\ref{prlb}}
\label{sec:lemma1proof}

\begin{proof}

Consider the input shown in Fig.~\ref{fig:lbexample}. If the algorithm turns on the generator at some point $s\in [1,\frac{\beta -w\delta_2}{\delta_1}]$, we can calculate the performance ratio as the online cost over the offline cost as follows:

\begin{eqnarray}
PR(s)=1+ \frac{\beta-(q(s+w)-q(s))+ max(q(s+w)- wc_m,0) }{    \frac{p_{\mathrm{max}}+\eta\cdot c_{\mathrm{g}}}{p_{\mathrm{max}}+\eta\cdot c_{\mathrm{g}}-c_{\mathrm{o}}} ((s+w)c_m + q(s+w) )    },
\end{eqnarray} 
where  $q(s)= \Delta(s)+\beta$.
Since we are looking for the lower bound, we find the minimum across all possible values of $s$. We call this value $R_{on}(\delta_1,\delta_2)$ and define it as follows:

\begin{eqnarray}
R_{on}(\delta_1,\delta_2)= \min_{s\in [1,\frac{\beta -w\delta_2}{\delta_1}]} PR(s).
\end{eqnarray}

By using the same logic in the proof of Theorem~\ref{lem:crfunction}, if the algorithm does not turn on the generator in $s\in [1,\frac{\beta -w\delta_2}{\delta_1}]$, and keeps the generator off, the  ratio  of the  online  cost  increment  over  the  offline  cost  increment  is 
\begin{eqnarray} 
 R_{\mathrm{off}}(\delta_2)=\frac{c_{\mathrm{m}}+\delta_2}{c_{\mathrm{m}}+\frac{Lc_{\mathrm{o}} \alpha  }{Lc_{\mathrm{o}}+c_{\mathrm{m}}}\delta_2}.
\end{eqnarray} 
Hence the lower bound can be calculated by finding the minimum of these two values: 
\begin{eqnarray} 
{\sf CR}_{\mathrm{l}}(\delta_1,\delta_2) = \min \{R_{\mathrm{on}}(\delta_1,\delta_2), R_{\mathrm{off}}(\delta_2)\}.
\end{eqnarray} 
This complete the proof of Lemma~\ref{prlb}.
\end{proof}
Similar to Theorem~\ref{lem:crfunction}, we want to find $(\delta_1^*,\delta_2^*)$ that maximizes the lower bound ${\sf CR}_{\mathrm{l}}(\delta_1,\delta_2)$. First, for each $\delta_2$, we find a corresponding $\delta_1$ such that $\delta_1= \text{arg}\max\limits_{ \delta }\, R_{\mathrm{on}}(\delta,\delta_2)$. This reduces $R_{\mathrm{on}}(\delta_1,\delta_2)$ to a single variable function of $\delta_2$.
\begin{subequations}\label{d1}
	\begin{eqnarray} 
		&& \delta_1= \text{arg}\max\limits_{ \delta }\, R_{on}(\delta,\delta_2) \\
		&\mbox{s.t.}& \delta_2 \leq \delta \leq (\beta - w\delta_2 )/w,\\
		&& \delta_2 \leq \delta \leq L \big(p_{\mathrm{max}}+\eta\cdot c_{\mathrm{g}}-c_{\mathrm{o}}-\frac{c_{\mathrm{m}}}{L}).
	\end{eqnarray}
\end{subequations} 
For a given $\delta_2$, the function $R_{on}(\delta_1,\delta_2)$ is a concave function of $\delta_1$. Therefore we can easily find a corresponding $\delta_1$ for each $\delta_2$.
Now both $R_{\mathrm{on}}$ and $R_{\mathrm{off}}$ are a function of $\delta_2$, and we find the maximum of the minimum of two single variable functions. We know that for $\delta_2=0$, we have $R_{\mathrm{on}}(\delta_1,0) \geq R_{\mathrm{off}}(0)=1 $, and $R_{\mathrm{off}}(\delta_2) $ is an increasing function. Hence, similar to \eqref{eq:finda} we keep increasing $\delta_2$ until we find the intersection of the two functions. Therefore, $\delta_2^*$ can be obtained by solving the following optimization problem:
\begin{subequations}\label{deltaa2}
	\begin{eqnarray}
		&& \delta_2^*= \text{arg}\max\limits_{ \delta_2 }\, {\sf CR}_{\mathrm{l}}(\delta_1,\delta_2)  \\
		&\mbox{s.t.}& 0 \leq \delta_2 \leq \beta/(2w),\\
		&& 0 \leq \delta_2\leq L \big(p_{\mathrm{max}}+\eta\cdot c_{\mathrm{g}}-c_{\mathrm{o}}-\frac{c_{\mathrm{m}}}{L}),\\
		&& R_{\mathrm{on}}(\delta_1,\delta_2) \geq  R_{\mathrm{off}}(\delta_2), \\
		&& \delta_1 \textit{is obtained from}~\eqref{d1}.
	\end{eqnarray}
\end{subequations} 
Therefore, we always have $R_{\mathrm{on}}(\delta_1^*,\delta_2^*) \geq R_{\mathrm{off}}(\delta_2^*)$ and ${\sf CR}_{\mathrm{l}}(\delta_1^*,\delta_2^*) = R_{\mathrm{off}}(\delta_2^*)$, which completes the proof of Theorem~\ref{CRLOB}.

\end{proof}

\end{document}